\documentclass[11pt,a4paper]{article}

\usepackage[T1]{fontenc}
\usepackage{lmodern}
\usepackage{microtype}
\usepackage{latexsym}
\usepackage{amsmath,amssymb,amsfonts,amsthm,mathtools}
\usepackage{mathrsfs}
\usepackage{graphicx}
\usepackage{float}
\usepackage[numbers,sort&compress]{natbib}
\usepackage{xcolor}
\usepackage[hidelinks]{hyperref}
\usepackage{comment}
\usepackage{subfigure}
\usepackage{float}

\hypersetup{
  pdftitle={Fredholm determinant solutions to mKP and 2DTL:
isospectral and non-isospectral reductions},
  pdfsubject={Fredholm determinants and integrable hierarchies},
  pdfkeywords={Fredholm determinant, tau-function, modified KP hierarchy, two-dimensional Toda lattice, modified Camassa-Holm, Bessel kernel}
}
\usepackage[nameinlink]{cleveref}
\usepackage[a4paper,left=0.82in,right=0.82in,top=0.72in,bottom=0.76in]{geometry}
\usepackage[normalem]{ulem}
\usepackage{soul}

\allowdisplaybreaks
\numberwithin{equation}{section}

\newtheorem{theorem}{Theorem}[section]
\newtheorem{lemma}[theorem]{Lemma}
\newtheorem{proposition}[theorem]{Proposition}
\newtheorem{corollary}[theorem]{Corollary}

\newtheorem{example}[theorem]{Example}
\newtheorem{remark}[theorem]{Remark}

\crefname{equation}{Eq.}{Eqs.}
\Crefname{equation}{Eq.}{Eqs.}
\crefname{section}{Sec.}{Secs.}
\Crefname{section}{Sec.}{Secs.}
\crefname{subsection}{Sec.}{Secs.}
\Crefname{subsection}{Sec.}{Secs.}
\crefname{subsubsection}{Sec.}{Secs.}
\Crefname{subsubsection}{Sec.}{Secs.}
\crefname{appendix}{App.}{Apps.}
\Crefname{appendix}{App.}{Apps.}
\crefname{theorem}{Thm.}{Thms.}
\Crefname{theorem}{Thm.}{Thms.}
\crefname{lemma}{Lem.}{Lems.}
\Crefname{lemma}{Lem.}{Lems.}
\crefname{proposition}{Prop.}{Props.}
\Crefname{proposition}{Prop.}{Props.}
\crefname{corollary}{Cor.}{Cors.}
\Crefname{corollary}{Cor.}{Cors.}
\crefname{definition}{Def.}{Defs.}
\Crefname{definition}{Def.}{Defs.}
\crefname{assumption}{Assump.}{Assumps.}
\Crefname{assumption}{Assump.}{Assumps.}
\crefname{example}{Ex.}{Exs.}
\Crefname{example}{Ex.}{Exs.}
\crefname{remark}{Rem.}{Rems.}
\Crefname{remark}{Rem.}{Rems.}

\newcommand{\dd}{\,\mathrm d}
\newcommand{\ii}{\mathrm i}
\newcommand{\Ai}{\operatorname{Ai}}
\newcommand{\Tr}{\operatorname{Tr}}

\newcommand{\R}{\mathbb R}

\title{Fredholm determinant solutions to the modified KP hierarchy and 2D Toda lattice:
isospectral and non-isospectral reductions}

\author{Dan Dai$^1$, Han-Han Sheng$^2$\\
Department of Mathematics, City University of Hong Kong, Hong Kong, P.R. China\\
  ${}^1$ \texttt{dandai@cityu.edu.hk}\\
    ${}^2$ \texttt{hsheng@cityu.edu.hk}
}
\date{}

\begin{document}
\maketitle

\begin{abstract}
We study Fredholm determinant solutions to the first two positive and the first negative members of the modified Kadomtsev-Petviashvili (mKP) hierarchy, together with the bilinear two-dimensional Toda lattice (2DTL) equation. By applying isospectral reductions, we explicitly construct Fredholm determinant solutions for the modified Korteweg-de Vries (mKdV), the modified Camassa-Holm, and the focusing nonlinear Schr{\"o}dinger equations. Furthermore, we demonstrate that in the non-isospectral case, the reduction from the mKP and 2DTL equations to the non-isospectral negative KdV equation yields solutions characterized by the deformed Bessel determinant.
\end{abstract}



\section{Introduction}\label{sec:introduction}

Determinantal point processes have attracted a lot of research interests in the past a few decades due to their connections with a wide range of topics in both mathematics and physics. For example, they are related to random unitary matrices \cite{And:Gui:Zei,mehta}, Dyson Brownian motion \cite{Dyson1962}, free fermionic theory \cite{Dean:LeD:Maj:Sch2016,Dou:Maj:Sch2018}, quantum gravity \cite{DiF:Gin:Zinn1995,Stan:Witt2020}, and many others. For a comprehensive overview of their properties and applications, we refer the reader to the excellent surveys by Soshnikov \cite{Sosh2000}, Johansson \cite{Johan2006}, Borodin \cite{Borodin-Survey2011}, and the references therein.

Let $\mathcal{X}$ be a configuration such that $\#(\mathcal{X} \cap [a,b])$ is finite for any bounded interval $[a,b] \subset\mathbb{R}$. A \emph{point process} $\mathcal{P}$ is a probability measure on the space of all the configurations. It is called \emph{determinantal} if, for each $k \in \mathbb{N}$, its $k$-point correlation function $\rho_k(x_1, \ldots, x_k) $ admits the determinant form
\begin{equation} \label{corre-det}
\rho_k(x_1, \ldots, x_k) = \det[K(x_i,x_j)]_{i,j=1}^k,
\end{equation}
where $K(\cdot,\cdot)$ is the \emph{correlation kernel}. A determinantal point process is completely characterized by its correlation kernel $K(\cdot,\cdot)$; indeed, all statistical information about the process is encoded in this kernel. Three canonical universal limiting kernels are the sine, Airy, and Bessel kernels, which describe the local statistics in the bulk, at the soft edge, and at the hard edge of the spectrum, respectively. They are given explicitly as follows:
\begin{align}
K_{\sin}(u,v) &=\frac{\sin \pi(u-v)}{\pi(u-v)}, \\
 K_{\mathrm{Ai}}(u,v) &=
\frac{\operatorname{Ai}(u)\operatorname{Ai}'(v)
-\operatorname{Ai}'(u)\operatorname{Ai}(v)}
{u-v}, \label{eq: Airy-kernel-def} \\
  K_{\mathrm{Be}}^{ (\alpha)}(u,v)&=\frac{J_{\alpha}\bigl(\sqrt{u}\bigr)\sqrt{v}\,J'_{\alpha}\bigl(\sqrt{v}\bigr) - \sqrt{u}\,J'_{\alpha}\bigl(\sqrt{u}\bigr)J_{\alpha}\bigl(\sqrt{v}\bigr)}{2(u - v)},\qquad u,v>0. \label{eq: Bessel-kernel-def}
\end{align}
In the above formulas, $\operatorname{Ai}$ is the Airy function and $J_\alpha$ is the Bessel function of the first kind of order $\alpha > -1$.

One of the central topics in the study of point processes is to understand the spacing distribution of the particles. This can be characterized by the so-called \emph{gap probability}, which is the probability that there are no particles in a given Borel set $J \subset \mathbb{R}$. For determinantal point processes, it is a fundamental result that the gap probability can be expressed as the Fredholm determinant of an integral operator:
\begin{align}
\mathbb P\{\text{no particles in } J\} &
=
\det(I-\mathcal K \vert_{L^2(J)} ) \label{eq:Fred-det-def}  \\
& =
1+\sum_{n=1}^{\infty}\frac{(-1)^n}{n!}
\int_{J^n}
\det\!\left[K(u_i,u_j)\right]_{i,j=1}^{n}
\,\mathrm{d}u_1\cdots \mathrm{d}u_n,
\end{align}
where $\mathcal K \vert_{L^2(J)} $ is the trace-class operator acting on $L^2(J)$ with kernel $K(\cdot,\cdot)$. Although Fredholm determinants are formally defined via the infinite series above, Jimbo et al. \cite{jimbo1980} discovered in their seminal work that the logarithmic derivative of the sine-kernel determinant on $[-s,s]$ satisfies an integrable nonlinear ordinary differential equation, namely the \emph{Painlev\'e}  V equation. Later, Tracy and Widom \cite{tracywidom1994airy,tracywidom1994bessel} adapted this approach to relate the Fredholm determinants of the Airy and Bessel kernels to the Painlev\'e II and III equations, respectively. In particular, the celebrated Tracy-Widom distribution is defined by the following Airy-kernel determinant:
\begin{equation}\label{Tracy-Widom formula}
    F_{ \mathrm{TW} }(s):= \det(I- \mathcal{K}_{\Ai} \vert_{(s, \infty)} ) := \exp \left(-\int_s^\infty (x-s) q_{\rm{HM}}^2(x) \dd x \right),
\end{equation}
where $\mathcal{K}_{\Ai} \vert_{(s, \infty)}$ is  the trace-class operator acting on $L^2 (s, \infty)$ with the Airy kernels $K_{\Ai}$ given in \eqref{eq: Airy-kernel-def}. Here, $q_{\mathrm{HM}}(x)$ is the Hastings-McLeod solution to the Painlev\'{e} II equation
\begin{equation}\label{Painleve II}
q''(x)=x q(x)+2 q^3(x),
\end{equation}
which is smooth on $\mathbb{R}$ and satisfies the asymptotic boundary condition $ q_{\rm{HM}}(x) \sim \operatorname{Ai} (x) $ as  $x\rightarrow+\infty. $

Beyond these connections to integrable ODEs, the relationship between Fredholm determinants and PDEs has also been investigated \cite{jimbo1980,tracywidom1994airy,tracywidom1994bessel}. In these works, when the Borel set $J$ in \eqref{eq:Fred-det-def} comprises a union of disjoint intervals rather than a single interval like that in \eqref{Tracy-Widom formula}, the determinant can be considered as a function of the endpoints of the interval. Its evolution with respect to these boundary points is shown to be governed by systems of integrable PDEs. A distinct perspective on PDE connections emerges in the context of finite-temperature models, as first demonstrated in the remarkable work of Its et al. \cite{its1990}. There, the Fredholm determinant of a finite-temperature sine kernel on a single interval $[-s,s]$ was shown to satisfy an integrable PDE. More recently, based on the framework of \cite{its1990}, Cafasso et al. \cite{cafassoclaeysruzza2021} studied the Fredholm determinant of the finite-temperature deformed Airy kernel, defined as
\begin{align} \label{eq:deformed-Airy-Fred-det}
K_{\sigma,x,t}^{\mathrm{Ai}}(u,v)
=
\sqrt{\sigma\left(\frac{u}{t^{2/3}}+\frac{x}{t} \right)} \, K_{\mathrm{Ai}}(u,v) \sqrt{\sigma\left(\frac{v}{t^{2/3}}+\frac{x}{t} \right)}, 
\qquad
 t>0,
\end{align}
where \(\sigma:\mathbb{R}\to[0,1]\) is an admissible weight satisfying certain smoothness and decay conditions.
Denoting by
\(\mathcal K_{\sigma,x,t}^{\mathrm{Ai}} \vert_{L^2(\mathbb{R})} \) the trace-class operator acting on $L^2(\mathbb{R})$ with this kernel, and defining the determinant
\begin{equation} \label{eq:finite-deformed-airy-det}
Q^{\mathrm{Ai}}_{\sigma}(x,t)
=
\det\!\left(
I-\mathcal K_{\sigma,x,t}^{\mathrm{Ai}} \vert_{L^2(\mathbb{R})}
\right),
\end{equation}
it was shown in \cite{cafassoclaeysruzza2021} that the function
\begin{align}
\phi(x,t)
=
\partial_x^2\log Q^{\mathrm{Ai}}_{\sigma}(x,t)
+
\frac{x}{2t},
\end{align}
solves the Korteweg-de Vries (KdV) equation
\begin{align} \label{eq: kdv-eqn}
\phi_t + 2 \phi \phi_x +
\frac{1}{6} \phi_{xxx}
=
0.
\end{align}
A special case of the admissible weight $\sigma(\cdot)$ in \eqref{eq:deformed-Airy-Fred-det} is the Fermi weight:
\begin{equation}\label{sigma-kpz}
\sigma_{\mathrm{KPZ}}(x)
=
\frac{1}{1+e^{-x}}. 
\end{equation}
For this specific choice, the deformed Airy-kernel determinant corresponds to the generating function of the solution to the Kardar-Parisi-Zhang (KPZ) equation with narrow-wedge initial data \cite{amircorwinquastel2011}. It is worth noting that this celebrated KPZ distribution is also characterized by an integro-differential generalization of the Painlev\'e II equation.

Building upon the relationship between exact KPZ statistics and one-dimensional integrable PDEs, a profound generalization was recently achieved by Quastel and Remenik \cite{quastelremenik2022}. They discovered that for a wide range of initial data, the logarithmic derivative of the generating function for KPZ solutions evolves according to the Kadomtsev-Petviashvili (KP) equation:
\begin{equation} \label{eq:KP-eqn}
\phi_t + \phi_r \phi + \frac{1}{12} \phi_{rrr} +  \frac{1}{4} \partial_r^{-1} \phi_{xx} = 0.
\end{equation}
If $\phi$ is independent of $x$, this equation reduces to the KdV equation \eqref{eq: kdv-eqn} up to a trivial rescaling. Because the integral kernels governing these general initial data are no longer the classical kernels or their finite-temperature deformations, Quastel and Remenik observed that they instead satisfy a set of underlying differential identities. Specifically, they proved that if $\mathcal{K}_{t,x,r}\vert_{L^2(\mathbb{R}^+)}$ is a trace-class operator acting on $L^2(\mathbb{R}^+)$ with a kernel $K(u,v; t,x,r)$ that is real analytic in $t$, $x$, and $r$, and satisfies the following differential relations
\begin{eqnarray}
K_r(u,v) &=& (\partial_u + \partial_v)K(u,v), \label{eq:QR-diff-condition1} \\
K_x(u,v) &= & (\partial_v^2 -  \partial_u^2)K(u,v), \label{eq:QR-diff-condition2} \\
K_t(u,v) &=& - \frac{1}{3} (\partial_u^3 + \partial_v^3)K(u,v), \label{eq:QR-diff-condition3}
\end{eqnarray}
then the logarithmic derivative $\phi(t,x,r) = \partial^2_{r} \log \det(I-\mathcal{K}_{t,x,r}) $ satisfies the KP equation \eqref{eq:KP-eqn}; see \cite[Theorem 1.3]{quastelremenik2022}\footnote{In fact, it was shown that the logarithmic derivative of the $n$-space point distribution functions of the KPZ fixed point satisfies the matrix KP equation, which involves an $n \times n$ matrix kernel acting on an $n$-fold direct sum space. This reduces to the scalar KP equation \eqref{eq:KP-eqn} when $n=1$.} and the earlier work by P\"oppe \cite{poppesattinger1988,poppe1989}. The connection between Fredholm determinants and KP equations has also been developed in the context of periodic growth models. Baik, Liu and Silva \cite{baik2022} showed that the logarithmic second derivative of the Fredholm determinant associated with periodic totally asymmetric simple exclusion process (TASEP) satisfies the KP equation.\footnote{Later, Baik, Prokhorov and Silva \cite{baik2023} also extended this framework to matrix KP equations.}

Beyond the KP hierarchy, there exist other fundamental classes of integrable systems, such as the modified Kadomtsev-Petviashvili (mKP) equation and the two-dimensional Toda lattice (2DTL) equation. These systems are mathematically rich, which couples continuous differential flows (including negative flows) with a discrete lattice variable. In this paper, we systematically study the Fredholm determinant solutions to both the mKP and the 2DTL equations. Motivated by the structural framework of Quastel and Remenik, our first main result provides sufficient conditions on a general kernel $K(u,v)$ under which the corresponding Fredholm determinants solve the bilinear mKP and 2DTL equations. Furthermore, just as the KdV equation \eqref{eq: kdv-eqn} can be understood as an isospectral reduction of the KP equation \eqref{eq:KP-eqn}, we investigate the corresponding isospectral reductions of the mKP and 2DTL equations. We also provide sufficient conditions on the kernel $K(u,v)$ such that the corresponding Fredholm determinants solve the modified Korteweg-de Vries (mKdV), the modified Camassa-Holm (mCH), and the focusing nonlinear Schr\"odinger (fNLS) equations. As our second main result, we investigate the \emph{non-isospectral} reduction of the mKP and 2DTL equations, which yields the negative KdV equation. Notably, this equation is directly connected to the PDE recently derived by Ruzza \cite{ruzza2025} for finite-temperature Bessel kernel determinants. 

The rest of the paper is organized as follows. In Sec. \ref{sec:main-theorem}, we state our main results about the Fredholm determinants solution to the  mKP and 2DTL hierarchies, as well as their isospectral and non-isospectral reductions. In Sec. \ref{sec:isospectral}, some explicit examples for the Fredholm determinant solutions of several classical (1+1)-dimensional integrable systems via isospectral reductions, are listed, including the mKdV equation and the mCH equation. We prove our main theorem in  Sec. \ref{sec:proof-main}. The non-isospectral reductions of the mKP and 2DTL equations and their relation to the Bessel-kernel determinant are provided in Sec. \ref{sec:nonisospectral}.  In App. \ref{pf-fnls}, we give an alternative Fredholm determinant representation for solutions of the fNLS equation, which we illustrate with examples of breather gas solutions. Finally, in  App. \ref{non-isospectral}, we establish non-isospectral Lax representations for the differential equations governing the finite-temperature deformed Bessel kernel determinant.

\section{Statement of results}\label{sec:main-theorem}

\subsection{The mKP-2DTL theorem and isospectral reductions}
\label{subsec:main-hierarchy-results}

We work on the half-line $\R_-=(-\infty,0]$, and consider the following four bilinear equations
 \begin{gather}(D_{x_1}^2+D_{x_2})\tau_n\cdot\tau_{n+1}=0,\label{eq:mkp-first}\\
 (D_{x_1}^3-3D_{x_1}D_{x_2}-4D_{x_3})\tau_n\cdot\tau_{n+1}=0,\label{eq:mkp-second}\\
 \bigl[D_{x_{-1}}(D_{x_1}^2+D_{x_2})-4D_{x_1}\bigr]
 \tau_n\cdot\tau_{n+1}=0,\label{eq:mkp-negative}\\
 \left(\frac12D_{x_1}D_{x_{-1}}-1\right)\tau_n\cdot\tau_n
 =-\tau_{n+1}\tau_{n-1}.\label{eq:2dtl-central} 
 \end{gather}
 Here Hirota's bilinear operators are defined by
\begin{equation}\label{eq:hirota-def}
 D_x^mD_t^nf\cdot g=
 (\partial_x-\partial_{x'})^m(\partial_t-\partial_{t'})^n
 f(x,t)g(x',t')\big|_{x'=x,\,t'=t}.
\end{equation}
The first three equations correspond to the first two positive flows and the first negative flow of the mKP hierarchy, respectively, while Eq. \eqref{eq:2dtl-central} is the bilinear
2DTL equation \cite{hirota2004,jimbomiwa1983,dickey2003}. Our first main result establishes Fredholm determinant solutions to these bilinear equations, as summarized in the following theorem.

\begin{theorem}[Fredholm determinant solutions to the mKP and 2DTL equations]\label{thm:main}
For each $n\in\mathbb Z$, let $\mathcal K^{(n)}(\mathbf x)$
be a trace-class integral operator on $L^2(\mathbb R_-)$,
with kernel $K^{(n)}(u,v;\mathbf x)$  and $\mathbf x = (x_{-1}, x_1, x_2, x_3)$.
Assume sufficient regularity and decay so that all kernel derivatives and primitives used below define trace-class operators, locally uniformly in $\mathbf x$. Suppose further that the kernel satisfies the following differential and shift relations:
\begin{align}
 \partial_{x_1} K^{(n)}(u,v)&=(\partial_u+\partial_v)K^{(n)}(u,v),\label{eq:flow-x1}\\
\partial_{x_2} K^{(n)}(u,v)&=(\partial_u^2-\partial_v^2)K^{(n)}(u,v),\label{eq:flow-x2}\\
 \partial_{x_3} K^{(n)}(u,v)&=(\partial_u^3+\partial_v^3)K^{(n)}(u,v),\label{eq:flow-x3}\\
 \partial_{x_{-1}}K^{(n)}(u,v)&=(\partial_u^{-1}+\partial_v^{-1})K^{(n)}(u,v),\label{eq:flow-xm1}\\
 \partial_u K^{(n)}(u,v)&+\partial_v K^{(n+1)}(u,v)=0,\label{eq:adjacent-order}
\end{align}
where $\partial_u^{-1}$ and $\partial_v^{-1}$ denote primitives normalized to vanish at $-\infty$. Then, the Fredholm determinant
\begin{equation}\label{eq:tau-main}
 \tau_n(\mathbf x)=\det(I+\mathcal{K}^{(n)}(\mathbf x)),
\end{equation}
satisfies the bilinear equations Eqs. \eqref{eq:mkp-first}-\eqref{eq:2dtl-central}.
\end{theorem}

\begin{remark}
It is well-known that the mKP and 2DTL hierarchies consist of an infinite sequence of commuting differential members; see \cite{hirota2004,jimbomiwa1983,dickey2003}. In the above theorem, we restrict our attention to the first two positive members (associated with the variables $x_1, x_2, x_3$) and the first negative member ($x_{-1}$) of the mKP hierarchy, alongside the central equation of the 2DTL hierarchy. By incorporating additional time variables into the parameter vector $\mathbf{x}$ and imposing the corresponding higher-order differential relations on the kernel $K^{(n)}(u,v;\mathbf{x})$, one can systematically extend this Fredholm determinant construction to higher-order members within these hierarchies.
\end{remark}

\begin{remark}
Thm. \ref{thm:main} generalizes the seminal results of Quastel and Remenik \cite{quastelremenik2022}. While they showed that the logarithmic derivative of the Fredholm determinant solves the nonlinear KP equation \cite[Theorem 1.3]{quastelremenik2022}, they also noted that the determinants themselves satisfy the equivalent Hirota bilinear equation \cite[Remark 1.2(4)]{quastelremenik2022}. We state our main theorem for the bilinear system (2.1)–(2.4), in which the Fredholm determinants enter directly as $\tau$-functions. This formulation provides a common starting point for the reductions considered below. However, in the subsequent corollaries concerning reductions to classical equations, where the nonlinear differential forms are concise and standard, we will revert to explicit differential equations.
\end{remark}

\begin{remark}
Quastel and Remenik's investigation was originally motivated by the $n$-space point distribution functions of the KPZ fixed point. Consequently, they established their results for the matrix KP equation, which involves an $n \times n$ matrix kernel acting on an $n$-fold direct sum space and reduces to the scalar KP equation \eqref{eq:KP-eqn} when $n=1$.  We also note that, in their study of the polynuclear growth model,
Matetski, Quastel, and Remenik \cite{MatetskiQuastelRemenik}
derived a matrix (non-Abelian) 2DTL structure from Fredholm determinants
associated with discrete matrix kernels acting on an $n$-fold
direct sum of $\ell^2(\mathbb Z_{>0})$. We expect that our results in Thm. \ref{thm:main} can be generalized to the matrix mKP and matrix 2DTL equations as well. However, for the sake of clarity and accessibility, we restrict ourselves to the scalar case.
\end{remark}

To illustrate Thm. \ref{thm:main}, we provide a concrete example of a kernel that satisfies the requisite flow equations.
\begin{example}
Let $\Gamma$ and $\widehat{\Gamma}$ be oriented contours, and $\omega(p,q)$ be a complex-valued weight function. For each $n\in\mathbb Z$, we define the kernel $K^{(n)}$ by
\begin{align}\label{eq:separated-model}
K^{(n)}(u,v;\mathbf x)=\int_\Gamma
 \int_{\hat{\Gamma}} \left(-\frac pq\right)^n e^{pu+\theta(p;\mathbf x)}e^{qv+\widehat\theta(q;\mathbf x)}\omega(p,q)\dd p\dd q,
\end{align}
where the phase functions are 
\begin{align*}
 \theta(p;\mathbf x)&=px_1+p^2x_2+p^3x_3+p^{-1}x_{-1},\\
 \widehat\theta(q;\mathbf x)&=qx_1-q^2x_2+q^3x_3+q^{-1}x_{-1}.
\end{align*}
We assume that $\Gamma$, $\widehat{\Gamma}$, and $\omega$ are chosen such that $K^{(n)}$ defines a trace-class integral operator, ensuring its Fredholm determinant is well-defined. By direct verification, $K^{(n)}$ satisfies the differential relations Eqs. \eqref{eq:flow-x1}-\eqref{eq:adjacent-order}. Consequently, the associated Fredholm determinant $\tau_n$ satisfies all four bilinear equations established in Thm. \ref{thm:main}.
\end{example}

As a consequence of Thm. \ref{thm:main}, we can obtain Fredholm determinant solutions for several classical (1+1)-dimensional integrable systems through reductions.

\begin{corollary}[Fredholm determinant solution to the mKdV equation]\label{thm:mkdv}
Let $\mathcal{K}=\mathcal{K}(x,t)$ be a trace-class integral operator on $L^2(\mathbb{R}_-)$ with kernel $K(u,v)$ satisfying
\begin{align}
& K_x=(\partial_u+\partial_v)K,\label{eq:mkdv-kernel-x}\\
 &K_t=-4(\partial_u^3+\partial_v^3)K,\label{eq:mkdv-kernel-t}\\
 &(\partial_u-\partial_v)K=0.\label{eq:mkdv-kernel-reduction}
\end{align}
Defining the $\tau$-functions as
\begin{equation}\label{eq:mkdv-tau-pair}
 \tau_+=\det(I+\mathcal{K}),\qquad \tau_-=\det(I-\mathcal{K}),
\end{equation}
then
\begin{equation}\label{eq:mkdv-field}
 \phi=\ii\,\partial_x\log\frac{\tau_+}{\tau_-}
\end{equation}
solves the mKdV equation
\begin{equation}\label{eq:mkdv}
 \phi_t+6\phi^2\phi_x+\phi_{xxx}=0.
\end{equation}
Furthermore, if the kernel $K(\xi,\eta)$ is purely imaginary, the resulting solution $\phi$ is real-valued wherever $\tau_+\tau_- \ne 0$.
\end{corollary}

\begin{corollary}[Fredholm determinant solution to the mCH equation]
\label{cor:mch}
Let $\mathcal K^{(n)}=\mathcal K^{(n)}(y,s)$, $n=0,1$, be
trace-class integral operators on $L^2(\mathbb{R}_-)$ with kernels
$K^{(n)}(u,v)$ satisfying 
\begin{align}
 &K_y^{(n)}
 = (\partial_u+\partial_v)K^{(n)},\label{mch-1}\\
 &K_s^{(n)}
 = \frac12\left(
 \partial_u^{-1}
 +\partial_v^{-1}
 \right)K^{(n)},\label{mch-2}\\
&\partial_u K^{(0)}
 +\partial_vK^{(1)}=0,\label{mch-3}\\
  &(\partial_u-\partial_v-1)K^{(n)}=0,\label{mch-4}
\end{align}
Define the $\tau$-functions as
\begin{equation}
 \tau_n=\det(I+\mathcal K^{(n)}),\qquad n=0,1.\label{tau-mch}
\end{equation}
Assume that $\tau_0,\tau_1>0$ and that $\tau_n-1$, together
with the required derivatives, tend to zero as $y\to-\infty$.
Then the parametric formula
\begin{equation}\label{hodo-mch}
 x=y+s+2\log\frac{\tau_1}{\tau_0},\qquad
 t=s,\qquad
 \phi=1-\partial_y\partial_s\log(\tau_0\tau_1)
\end{equation}
defines a solution of the mCH equation
\begin{equation}\label{eq:mch}
 m_t+\bigl[m(\phi^2-\phi_x^2)\bigr]_x=0,\qquad m=\phi-\phi_{xx},
\end{equation}
wherever $x_y\ne0$.
\end{corollary}

In Sec. \ref{sec:isospectral}, we provide several explicit Fredholm determinant solutions for the equations in the above corollaries. These include soliton and soliton gas solutions for the mKdV equation in Examples \ref{ex:mkdv-one-soliton} and \ref{ex:mkdv-hankel}, a soliton gas solution for the mCH equation in Example \ref{ex:mch-separated}.

\subsection{Non-isospectral reductions: the Bessel-kernel determinant}
\label{subsec:main-nonautonomous-results}

Our second main result is motivated by the relationship of the KPZ narrow wedge solution to the KP and KdV equations. Recall that the generating function of the KPZ narrow wedge solution admits the following Fredholm determinant representation (cf. \cite{amircorwinquastel2011})
\begin{align} \label{eq:nw-KPZ-Fred}
    G_{\text{nw}}(r,x,t)=\det (I-\mathcal{K}_{\text{nw}}|_{L^2(\mathbb{R}_+)}),
\end{align}
where the kernel of $\mathcal{K}_{\text{nw}}$ is given by
\begin{align}
K_{\text{nw}}(u,v)=\int_\R\frac{t^{-2/3}}{1+e^y}
\Ai\!\left(t^{-1/3}(u+r-y)+t^{-4/3}x^2\right)\Ai\!\left(t^{-1/3}(v+r-y)+t^{-4/3}x^2\right)\dd y.\label{nwKPZ-ker}
\end{align}
Conjugating this kernel by the factor $e^{(v-u)x/t}$ leaves the underlying Fredholm determinant unchanged. Because the resulting conjugated kernel $e^{(v-u)x/t}K_{\text{nw}}(u,v)$ satisfies the differential conditions Eqs. \eqref{eq:QR-diff-condition1}-\eqref{eq:QR-diff-condition3}, Quastel and Remenik established in \cite[Thm. 3.1]{quastelremenik2022} that $G_{\text{nw}}(r,x,t)$ serves as a $\tau$-function for the KP equation \eqref{eq:KP-eqn}. Furthermore, they also suggested the following coordinate transformation
\begin{align} \label{eq:QR-change-of-variables}
    X=r+\frac{x^2}{t},\ T=t,
\end{align}
which maps the kernel $K_{\text{nw}}$ to
\begin{align}
   \widetilde{K}_{\text{nw}}(u,v)&=\int_\R T^{-2/3}\frac{1}{1+e^y}
\Ai\!\left(T^{-1/3}(u+X-y)\right)\Ai\!\left(T^{-1/3}(v+X-y)\right)\dd y\notag\\
&=\int_\R T^{-1/3}\sigma_{\mathrm{KPZ}}(-X+T^{\frac 13}z)
\Ai\!\left(z+T^{-1/3}u\right)\Ai\!\left(z+T^{-1/3}v\right)\dd z.
\end{align}
Here, $\sigma_{\mathrm{KPZ}}(\cdot)$ is the Fermi weight defined in \eqref{sigma-kpz}. Consequently, the Fredholm determinant associated with this transformed kernel coincides exactly with $Q^{\mathrm{Ai}}_{\sigma}(x,t)$ as defined in \eqref{eq:finite-deformed-airy-det}. 

To see this, we factor the associated integral operator as $\widetilde{\mathcal{K}}_{\text{nw}}\vert{}_{L^2(\mathbb{R}_+)}=\mathcal{A}^* \circ \mathcal{A}$, where $\mathcal{A}:\ L^2(\mathbb{R}_+)\rightarrow L^2(\mathbb{R})$ is defined as  
\begin{align}
    (\mathcal{A}h)(z):=T^{-\frac 16}\sqrt{\sigma_{\mathrm{KPZ}}(-X+T^{\frac 13}z)}\int_0^{+\infty}\Ai\!\left(z+T^{-1/3}u\right)h(u)\dd u,
\end{align}
and its adjoint $\mathcal{A}^*: \ L^2(\mathbb{R})\rightarrow L^2(\mathbb{R}_+)$ is given by
\begin{align}
    (\mathcal{A}^*f)(u):=T^{-\frac 16}\int_\mathbb{R} \sqrt{\sigma_{\mathrm{KPZ}}(-X+T^{\frac 13}z)}\Ai\!\left(z+T^{-1/3}u\right)f(z)\dd z.
\end{align}
A direct norm calculation shows that $\mathcal{A}$ is a Hilbert-Schmidt operator:
 \begin{align}
     ||\mathcal{A}||_2^2&=\int_\mathbb{R}\int_0^{+\infty}T^{-\frac 13}\sigma_{\mathrm{KPZ}}(-X+T^{\frac 13}z)\Ai^2\!\left(z+T^{-1/3}u\right)\dd u \dd z\notag\\
     &=\int_\mathbb{R}\sigma_{\mathrm{KPZ}}(-X+T^{\frac 13}z)\int_0^{+\infty}\Ai^2\!\left(z+a\right)\dd a \dd z\notag\\
     &=\int_\mathbb{R}\sigma_{\mathrm{KPZ}}(-X+T^{\frac 13}z)K_{\Ai}(z,z)\dd z < +\infty.
 \end{align}
This gives us
\begin{equation}
      \widetilde G_{\text{nw}}(X,T)=\det (I-\widetilde{\mathcal{K}}_{\text{nw}}|_{L^2(\mathbb{R}_+)}) = \det (I-\mathcal{A}^* \circ \mathcal{A}|_{L^2(\mathbb{R}_+)}) = \det (I-\mathcal{A}\circ \mathcal{A}^*|_{L^2(\mathbb{R})}),
\end{equation}
where the kernel of the composition $\mathcal{A} \circ \mathcal{A}^*$ on $L^2(\mathbb{R})$ is given by
\begin{align}
   A\circ A^*(u,v)&=T^{-\frac 13}\sqrt{\sigma_{\mathrm{KPZ}}(-X+T^{\frac 13}u)}\sqrt{\sigma_{\mathrm{KPZ}}(-X+T^{\frac 13}v)}\int_0^{+\infty}\Ai\!\left(u+T^{-1/3}a\right)\Ai\!\left(v+T^{-1/3}a\right)\dd a\notag\\
   &=\sqrt{\sigma_{\mathrm{KPZ}}(-X+T^{\frac 13}u)}\sqrt{\sigma_{\mathrm{KPZ}}(-X+T^{\frac 13}v)}\int_0^{+\infty}\Ai\!\left(u+b\right)\Ai\!\left(v+b\right)\dd b\notag\\
   &=\sqrt{\sigma_{\mathrm{KPZ}}(-X+T^{\frac 13}u)}\sqrt{\sigma_{\mathrm{KPZ}}(-X+T^{\frac 13}v)}K_{\Ai}(u,v) = K_{\sigma,\rho,s}^{\mathrm{Ai}}(u,v),
\end{align}
with  $\rho=-XT^{-\frac 12},\ s=T^{-\frac 12}$. Note that $K_{\sigma,\rho,s}^{\mathrm{Ai}}(u,v)$ is the finite-temperature deformed Airy kernel given in \eqref{eq:deformed-Airy-Fred-det}, whose Fredholm determinant is related to the KdV equation; see Eqs. \eqref{eq:finite-deformed-airy-det}-\eqref{eq: kdv-eqn}.

Consequently, the coordinate transformation \eqref{eq:QR-change-of-variables} establishes a direct relation between the Fredholm determinant solutions of the KP and KdV equations. At the level of the associated partial differential equations, this relationship can be understood as a non-isospectral reduction. We formulate this connection in the following proposition.
\begin{proposition} \label{prop:KP-reduction}
    Let $\phi(r,x,t)=\partial_r^2\log\tau(r,x,t)$ satisfy the KP equation \eqref{eq:KP-eqn}, and $\tau(r,x,t)$ satisfy the bilinear KP equation. Then, under the following non-isospectral reduction
\begin{align}
    (t\partial_x-2x\partial_r)\tau=0,\label{ckdv-red}
\end{align}
and eliminating the variable $x$, we have $\phi(r,0,t)$ satisfies the cylindrical KdV equation
\begin{align}
\phi_t+\phi\phi_r+\frac{1}{12}\phi_{rrr}+\frac{\phi}{2t}=0.\label{ckdv}
\end{align}
\end{proposition}
\begin{proof}
The Hirota bilinear form of the KP equation \eqref{eq:KP-eqn} is
\[
\left(D_rD_t+\frac1{12}D_r^4+\frac14D_x^2\right)
\tau\cdot\tau=0.
\]
The constraint in Eq. \eqref{ckdv-red} gives
\[
\tau_x\big|_{x=0}=0,
\qquad
\tau_{xx}\big|_{x=0}
=\frac{2}{t}\tau_r\big|_{x=0}.
\]
Hence, combining the above formulas and restricting the 
to $x=0$, we have
\[
\left(D_rD_t+\frac1{12}D_r^4\right)
\widehat\tau\cdot\widehat\tau
+\frac1{2t}\partial_r(\widehat\tau^2)=0,
\qquad
\widehat\tau(r,t)=\tau(r,0,t).
\]
Dividing by $2\widehat\tau^2$ and taking the derivative with respect to $r$ directly yields Eq. \eqref{ckdv} for
$\phi(r,0,t)=\partial_r^2\log\widehat\tau$.
\end{proof}
\begin{remark}
    Under the variable transformation
    \begin{align}
        \rho=-rt^{-\frac 12},\ s=t^{-\frac 12},\ \tilde{\phi}(\rho,s)=t\phi(r,t)-\frac 12 r,
    \end{align}
    the cylindrical KdV equation \eqref{ckdv} transforms directly into the standard KdV equation \eqref{eq: kdv-eqn}; see \cite{cafassoclaeysruzza2021,claeysglesnerruzzatarricone2024}.
\end{remark}

Returning to the Bessel kernel defined in \eqref{eq: Bessel-kernel-def}, we recall its double integral representation (cf. \cite[Prop. 2.1]{girotti2015}):
\begin{equation} \label{eq: Bessel-kernel-double-integal}
K_{\mathrm{Be}}^{ (\alpha)}(u,v) =  \left(\frac{v}{u} \right)^{\alpha/2} \int_\gamma \frac{\dd q}{2\pi i}\ \int_{\hat\gamma} \frac{\dd p}{2\pi i} \frac{e^{u q - \frac{1}{4 q} - v p + \frac{1}{4 p}}}{p - q} \left(\frac{p}{q}\right)^{\alpha} ,\quad u,v>0,
\end{equation}
where $\gamma$ is a curve that extends to $-\infty$ and winds counterclockwise around the origin, and $\hat \gamma$ is its image under the mapping $q \mapsto 1/q$. 

Comparing this representation with \eqref{eq:separated-model}, it is natural to expect that the Bessel kernel belongs to the family, and that its corresponding Fredholm determinant is related to the $\tau$-functions of the mKP or 2DTL hierarchies established in Thm. \ref{thm:main}. The connection between the Fredholm determinant $\det(I- \mathcal{K}_{\mathrm{Be}} \vert_{(0,s)})$ and the Toda lattice was originally discovered by Forrester and Witte \cite[Sec. 4]{ForresterWitte2002}. Here, we naturally recover this relationship within our framework.
\begin{proposition}\label{prop:cly-toda}
Let $G_\alpha(x_{-1})=\det(I-\mathcal{K}_{\mathrm{Be}} ^{(\alpha)} |_{L^2(0,-4x_{-1})})$ with $x_{-1}<0$. Then, $G_\alpha$  satisfies the Toda lattice equation
\begin{align} 
    \delta^2\log G_\alpha=x_{-1}\left(1-\frac{G_{\alpha+1}G_{\alpha-1}}{G_{\alpha}^2}\right),\quad \delta=x_{-1}\partial_{x_{-1}}.\label{eqn:cyl-toda}
\end{align}
\end{proposition}
\begin{remark}
The equation in the preceding proposition is equivalent to the one in \cite[Eq. (4.12)]{ForresterWitte2002} with $x_{-1}=-t$, $\bar{G}_\alpha=t^{\frac{\alpha^2}{2}}e^tG_{\alpha}$. 
\end{remark}

The detailed proof of this proposition is given in Sec. \ref{sec:nonisospectral}.

Recently, \citet{ruzza2025} studied the finite-temperature deformed Bessel kernel
\begin{equation}
K_{\sigma,x,t}^{\mathrm{Be}}(u,v)
=
\sigma\!\left(r\right)
K_{\mathrm{Be}}^{(\alpha)}(u,v),
\quad r=x^{-2}u+t,\quad
x>0.   
\end{equation}
For a large class of admissible weight $\sigma(\cdot)$, he established the analytic framework for the  corresponding Fredholm determinant 
\begin{align} \label{eq:Bessel-deformed-det}
&Q^{\mathrm{Be}}_{\alpha,\sigma}(x,t)
=
\det\!\left(
I-\mathcal K_{\sigma,x,t}^{\mathrm{Be}} \vert_{L^2(0,\infty)}
\right)
\end{align}
More precisely, let 
\begin{align}
\phi(x,t)
=
-\partial_x\log Q^{\mathrm{Be}}_{\alpha,\sigma}(x,t)
+\frac{xt}{2}
-\frac{4\alpha^2-1}{8x},
\end{align}
then \(\phi(x,t)\) satisfies the following nonlinear integrable PDE (cf. \cite[Thm 1.2]{ruzza2025})
\begin{align}
    (2\phi_x-t)\phi_t^2
+\frac14 \phi_{xt}^2
-\frac12 \phi_{xxt}\phi_t
=
\frac{\alpha^2}{4}.\label{n-nkdv}
\end{align}
Differentiating the above equation with respect to $x$, the above equation becomes
\begin{align}
 \phi_{xx}\phi_t+(2\phi_x-t)\phi_{xt}-\frac{1}{4}\phi_{xxxt}
=0.\label{n-nkdv1}
\end{align}
Ruzza \cite{ruzza2025} also briefly noted a structural similarity between the aforementioned equations and the negative KdV equation. In the present paper, we demonstrate that Eq. \eqref{n-nkdv1} can be derived directly from the coupled mKP-2DTL system \eqref{eq:mkp-first}-\eqref{eq:2dtl-central} via a non-isospectral reduction. This derivation closely parallels the methodology employed in Prop. \ref{prop:KP-reduction}, where a similar non-isospectral constraint was used to extract the cylindrical KdV equation from the KP framework.

\begin{theorem}\label{thm:bessel-reduction}
    Let $\tau_n(x_1,x_2,x_{-1})$ satisfy the first mKP equation \eqref{eq:mkp-first} and the 2DTL equation \eqref{eq:2dtl-central}. Suppose that $\tau_n$ obeys the non-isospectral reduction constraint evaluated at $x_2=0$:
 \begin{align}
        \left(x_1\partial_{x_2}+(n+\alpha)\partial_{x_1}\right)\tau_n=0.\label{noniso-reduction}
\end{align}
    Then, under the variable transformation $x_1=\frac{x^2}{4}$ and $x_{-1}=t$, the function 
    \begin{align}
        \phi(x,t)=-\partial_x\log \tau_0(\frac{x^2}{4},0,t)
+\frac{xt}{2}
-\frac{4\alpha^2-1}{8x},
    \end{align}
    satisfies Eq. \eqref{n-nkdv1}.
\end{theorem}

\begin{remark}
The Fredholm determinant of the deformed Bessel kernel \eqref{eq:Bessel-deformed-det} serves as a canonical example of this reduction.  Furthermore, in App. \ref{non-isospectral}, we demonstrate that applying a non-isospectral deformation to the Lax pair of the negative KdV equation yields a new Lax pair that directly generates Eq. \eqref{n-nkdv1}. This construction provides a clear and exact explanation for Ruzza's observation.
\end{remark}
The detailed proof is given in Sec. \ref{sec:nonisospectral}.

\section{Isospectral reductions and Fredholm determinant solution to some (1+1)-dimensional integrable systems}\label{sec:isospectral}

In this section, we aim to obtain solutions to some classical (1+1)-dimensional integrable systems in the form of Fredholm determinants through isospectral reductions.

\subsection{Classical KP reductions}\label{subsec:classical-kp}
We first record the result of the KP equation in the normalization used below.

\begin{lemma}[\cite{poppe1989,poppesattinger1988,quastelremenik2022}]
\label{lem:poppe-kp}
Let $\mathcal{K}=\mathcal{K}(x_1,x_2,x_3)$ be a trace-class integral operator on $L^2(\R_-)$
with sufficiently
decaying kernel $K(u,v)$.  If
\begin{equation}\label{eq:poppe-kp-flows}
 K_{x_j}=\bigl(\partial_u^j-(-\partial_v)^j\bigr)K,
 \qquad j=1,2,3,
\end{equation}
then $\tau=\det(I+\mathcal{K})$ satisfies
\begin{equation}\label{eq:bilinear-kp}
 (D_{x_1}^4-4D_{x_1}D_{x_3}+3D_{x_2}^2)\tau\cdot\tau=0.
\end{equation}
\end{lemma}

Two classic reductions of the KP equation are 2-reduction to the KdV equation and 3-reduction to the Boussinesq equation. Here $p$-reduction means making the equation independent of the variable $x_p$. Thus, from Lem. \ref{lem:poppe-kp} we can obtain the following corollary.

\begin{corollary}
    [Fredholm determinant solution for KdV and Boussinesq]\label{prop:classical-kernel-criteria}
Let $\mathcal{K}$ be a trace-class integral operator with a sufficiently
decaying kernel $K(u,v)$.

\smallskip
\noindent\emph{(i) KdV (also see in \cite{poppe1984}).}  Suppose that $\mathcal{K}=\mathcal{K}(X,T)$ satisfies
\begin{align}
 &K_X=(\partial_u+\partial_v)K,\label{eq:kdv-kernel-x}\\
 &K_T=-4(\partial_u^3+\partial_v^3)K,\label{eq:kdv-kernel-t}\\
& (\partial_u-\partial_v)K=0.\label{eq:kdv-kernel-reduction}
\end{align}
On every open set where $\det(I+\mathcal{K})\ne0$, the function
\begin{equation}\label{eq:kdv-field}
\phi(X,T)=2\partial_X^2\log\det(I+\mathcal{K})
\end{equation}
solves
\begin{equation}\label{eq:kdv}
 \phi_T+6 \phi \phi_X+ \phi_{XXX}=0.
\end{equation}

\smallskip
\noindent\emph{(ii) Boussinesq.}  Suppose that $\mathcal{K}=\mathcal{K}(X,T)$ satisfies
\begin{align}
 &K_X=(\partial_u+\partial_v)K,\label{eq:bq-kernel-x}\\
 &K_T=(\partial_u^2-\partial_v^2)K,\label{eq:bq-kernel-t}\\
 &(\partial_u^2-\partial_u\partial_v+\partial_v^2)K=0.\label{eq:bq-kernel-reduction}
\end{align}
On every open set where $\det(I+\mathcal{K})\ne0$, the function
\begin{equation}\label{eq:bq-field}
 \phi(X,T)=3\partial_X^2\log\det(I+\mathcal{K})
\end{equation}
solves
\begin{equation}\label{eq:boussinesq}
 3\phi_{TT}+2(\phi^2)_{XX}+\phi_{XXXX}=0.
\end{equation}
\end{corollary}

\begin{proof}
For part~(i), Eq. \eqref{eq:kdv-kernel-reduction} is the reduction condition that makes the \(x_2\)-flow in Lem. \ref{lem:poppe-kp} vanishes,
while Eq. \eqref{eq:kdv-kernel-t} corresponds to $x_3=-4T$. In this case, the bilinear equation in
Lem. \ref{lem:poppe-kp} becomes 
\begin{align}
    \left(D_X^4+D_XD_T\right)\tau\cdot\tau=0,
\end{align}
which is nothing but the bilinear KdV equation. 

As we mentioned before, Part~(ii) is the
3-reduction of KP. Eq. \eqref{eq:bq-kernel-reduction} removes the $x_3$ flow,
and Eqs. \eqref{eq:bq-kernel-x}-\eqref{eq:bq-kernel-t} identify the remaining two flows. The bilinear equation \eqref{eq:bilinear-kp} are rewritten as
\begin{align}
    \left(D_X^4+3D_T^2\right)\tau\cdot \tau=0
\end{align}
The dependent variable transformation in Eq. \eqref{eq:kdv-field} and Eq. \eqref{eq:bq-field} gives the nonlinear equations.
\end{proof}

\subsection{Fredholm determinant solution to the mKdV equation}\label{subsec:mkdv}

We first prove Cor. \ref{thm:mkdv}, and then give its
one-soliton and continuous spectral realizations.

\begin{proof}[Proof of Cor. \ref{thm:mkdv}]
Based on Thm. \ref{thm:main}, we start from Eq. \eqref{eq:mkp-first} and Eq. \eqref{eq:mkp-second}. Similar to the case of the KdV equation, Eq. \eqref{eq:mkdv-kernel-reduction} is the 2-reduction condition, and Eq. \eqref{eq:mkdv-kernel-x} and Eq. \eqref{eq:mkdv-kernel-t} correspond to the variable transformation $x_1=x,\ x_3=-4t$. By setting $K^{(n)}=(-1)^nK$, Eq. \eqref{eq:adjacent-order} holds automatically. Then with the definition of $\tau_+$ and $\tau_-$ in Eq. \eqref{eq:mkdv-tau-pair}, we have
\begin{align}
 D_x^2\tau_+\cdot\tau_-&=0,\label{eq:mkdv-bilinear-1}\\
 (D_x^3+D_t)\tau_+\cdot\tau_-&=0.\label{eq:mkdv-bilinear-2}
\end{align}
Where $\tau_+\tau_-\ne0$, these identities imply
\[
 \ell_t+\ell_{xxx}-2\ell_x^3=0,
 \qquad
 \ell=\log\frac{\tau_+}{\tau_-}.
\]
Differentiating in $x$ and setting $\phi=i\ell_x$
gives the mKdV equation.
If $K$ is purely imaginary, then
$\tau_-=\overline{\tau_+}$, so $\phi$ is real.
\end{proof}

\begin{example}[One-soliton solution]\label{ex:mkdv-one-soliton}
Let $\kappa>0$ and $x_0\in\R$.  Define the rank-one operator
$\mathcal{K}$ on $L^2(\mathbb{R_-})$ with kernel $K(u,v)$
\begin{equation}\label{eq:mkdv-soliton-kernel}
 K(u,v;x,t)
 =-2i\kappa
   e^{\kappa(u+v+2x-2x_0)-8\kappa^3t}.
\end{equation}
It is trace class, purely imaginary, and satisfies the conditions of
Cor. \ref{thm:mkdv}.  This kernel yields one-soliton solution to the mKdV equation
\begin{equation}\label{eq:mkdv-soliton}
 \phi(x,t)
 =2\kappa\operatorname{sech}
   \bigl(2\kappa(x-4\kappa^2t-x_0)\bigr).
\end{equation}
\end{example}

\begin{example}[The dense mKdV soliton gas]
\label{ex:mkdv-hankel}
Fix $0<\eta_1<\eta_2$, let $r$ be the positive spectral density used in
\cite{girotti2023}, and set
\begin{equation}\label{eq:mkdv-gas-profile}
 K(u,v;x,t)
 =-\frac{i}{2\pi}\int_{\eta_1}^{\eta_2}
 r(ip)e^{p(u+v+2x)-8p^3t}\,\dd p.
\end{equation}
This is the half-line form of \citet[Eq.~(3.14)]{girotti2023}.  It is trace-class, purely imaginary, and satisfies 
Eqs. \eqref{eq:mkdv-kernel-x}-\eqref{eq:mkdv-kernel-reduction}.  Consequently,
\begin{equation}\label{eq:mkdv-gas-solution}
\phi(x,t)=\ii\partial_x\log
 \frac{\det(I+\mathcal K)}{\det(I-\mathcal K)}
\end{equation}
is the solution to the mKdV equation. Equivalently, by
Sylvester's identity, the same determinants are those of the Cauchy operator
\(\mathcal C_{\rm mKdV}\) on $\Sigma_1=\ii(\eta_1,\eta_2)$ with kernel
\begin{equation}\label{eq:mkdv-girotti-contour-kernel}
 c_{\rm mKdV}(u,v)=
 \frac{\sqrt{r(u)}e^{-\ii\theta(u;x,t)}
       \sqrt{r(v)}e^{-\ii\theta(v;x,t)}}
      {2\pi\ii(u+v)},
 \qquad \theta(u;x,t)=4tu^3+xu,
\end{equation}
used in \citet[Thm.~2.1]{girotti2023}. This is the Fredholm determinant form they derived for the dense soliton gas. A similar form can be found in \cite{tracywidom1996mkdv}.
\end{example}

\subsection{Fredholm determinant solution to the mCH equation}\label{subsec:mch}
We prove Cor. \ref{cor:mch} using the first positive
and first negative mKP members, i.e., Eq. \eqref{eq:mkp-first} and \eqref{eq:mkp-negative}.

\begin{proof}[Proof of Cor. \ref{cor:mch}]
Introduce the change of variables
\begin{align}
x_1=y+T,\qquad
x_2=T,\qquad
x_{-1}=\frac{s}{2},
\end{align}
then we have
\begin{align}
\partial_{x_1}=\partial_y,
\qquad
\partial_{x_2}=\partial_T-\partial_y,
\qquad
\partial_{x_{-1}}=2\partial_s.
\end{align}
Therefore, kernel identities \eqref{mch-1}-\eqref{mch-3} correspond to the variable transformation. Reduction condition \eqref{mch-4} means
\begin{align}
    \partial_T K^{(n)}=(\partial_{x_1}+\partial_{x_2})K^{(n)}=(\partial_u+\partial_v)(\partial_u-\partial_v+1)K^{(n)}=0.
\end{align}
Then, with the definition of $\tau_n$ in Eq. \eqref{tau-mch}, we have
\begin{align}
 &(D_y^2+D_y)\tau_1\cdot\tau_0=0,
 \label{eq:mch-reduced-positive}\\
 &\bigl(2D_sD_y^2+2D_sD_y-4D_y\bigr)
 \tau_1\cdot\tau_0=0,
 \label{eq:mch-reduced-negative}
\end{align}
which are the bilinear equations of the mCH equation \eqref{eq:mch}. Under the transformation \eqref{hodo-mch}, they can be restored to the mCH equation; see \citet{huyinwu2016,fenghu2024}.
\end{proof}

\begin{example}
\label{ex:mch-separated}
Let $\Gamma=[p_-,p_+]\subset(0,\frac 12)\cup(\frac 12,+\infty)$ and let $\mu$ be a finite
positive measure on $\Gamma$.  Define
\begin{align}\label{eq:mch-separated-kernel}
K^{(n)}(u,v;y,s)
 =\int_\Gamma
 \left(-\frac{2p-1}{2p+1}\right)^n
 \exp\left(
p(u+v+2y)+\frac{4p}{4p^2-1}s
 \right)\,\dd\mu(p),
 \qquad n=0,1.
\end{align}
Direct differentiation verifies the kernel $e^{-\frac12 (u-v)}K^{(n)}(u,v;y,s)$ satisfies
identities in Cor. \ref{cor:mch}. Since the factor $e^{-\frac12 (u-v)}$ leaves the associated Fredholm determinant invariant, we have $\tau_n=\det(I+\mathcal{K}^{(n)})$ under the transformation \eqref{hodo-mch} is the solution of the mCH equation.

This Fredholm determinant has a Cauchy-operator form
parallel to the mKdV soliton gas.  Indeed, defining
\begin{equation}\label{eq:mch-spectral-weight}
 \omega_n^{\rm mCH}(p;y,s)=
 \left(-\frac{2p+1}{2p-1}\right)^{-n}
 \exp\left(2py+\frac{4p}{4p^2-1}s\right),
\end{equation}
Sylvester's identity gives
\begin{equation}\label{eq:mch-cauchy-determinant}
 \det(I+\mathcal{K}^{(n)})_{L^2(\mathbb{R_-})}
 =\det(I+\mathcal C_n^{\rm mCH})_{L^2(\Gamma,\dd\mu)},
 \quad
 c_n^{\rm mCH}(p,q)=
\frac{
\sqrt{\omega_n^{\mathrm{mCH}}(p;y,s)}
\sqrt{\omega_n^{\mathrm{mCH}}(q;y,s)}
}{p+q},
\end{equation}
where $\mathcal C_n^{\rm mCH}$ is the integral operator with kernel
$c_n^{\rm mCH}$.
For a discrete measure $\mu=\sum_{j=1}^N w_j\delta_{p_j}$, the determinants
reduce to
\[
 \tau_n=
 \det_{1\leq i,j\leq N}
 \left[
 \delta_{ij}
 +\frac{w_i}{p_i+p_j}
 \left(-\frac{2p_i-1}{2p_i+1}\right)^n
 \exp\left(2p_i y+\frac{4p_i}{4p_i^2-1}s\right)
 \right],
\]
which gives the finite Gram determinants in
\cite{huyinwu2016,fenghu2024}. For $\dd\mu(p)=r(p)\,\dd p$ with continuous $r\ge0$, we can obtain the dense soliton gas to the mCH equation. Fig. \ref{mch-solitongas} shows the dynamics of soliton gas, as well as the acceleration of an individual soliton by the soliton gas.
\end{example}

\begin{remark}
    Under the transformation $\lambda=p^2,\ \mu=q^2$, the kernel $c_n^{\mathrm{mCH}}$ in Eq. \eqref{eq:mch-cauchy-determinant} can be rewritten into an integrable form in the sense of Its-Izergin-Korepin-Slavnov (IIKS) \cite{its1990}. This structure makes it possible to rigorously study the asymptotic behavior of the associated solutions via Riemann-Hilbert techniques. We intend to explore this asymptotic analysis in a future publication.
\end{remark}

\begin{figure}[H]
	\centering
	\subfigure[pure soliton gas]
    {
	    \includegraphics[width=2.5in]{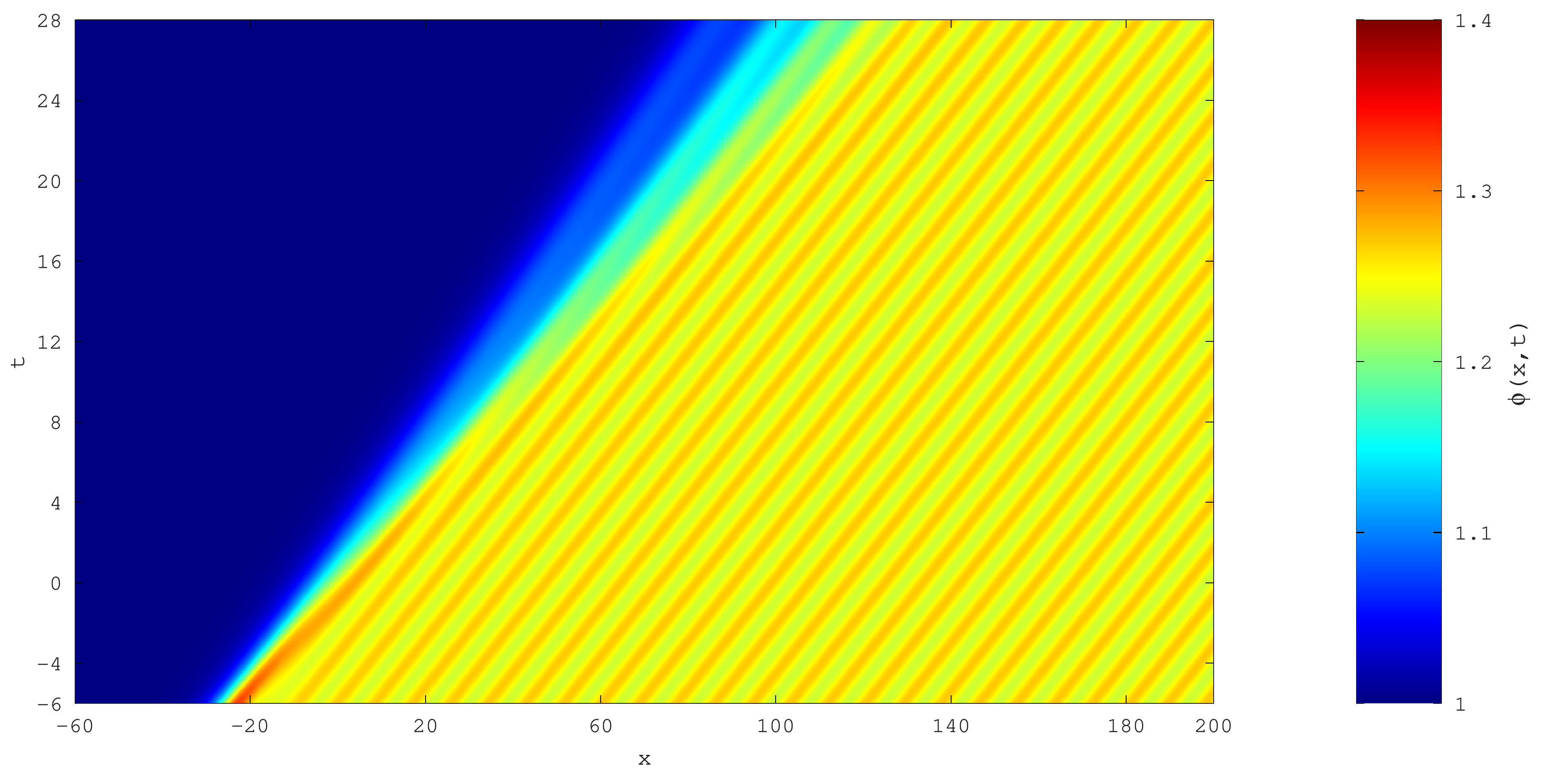}
    }
	\hspace*{3em}
    	\subfigure[pure soliton gas]
	{
		\includegraphics[width=2.5in]{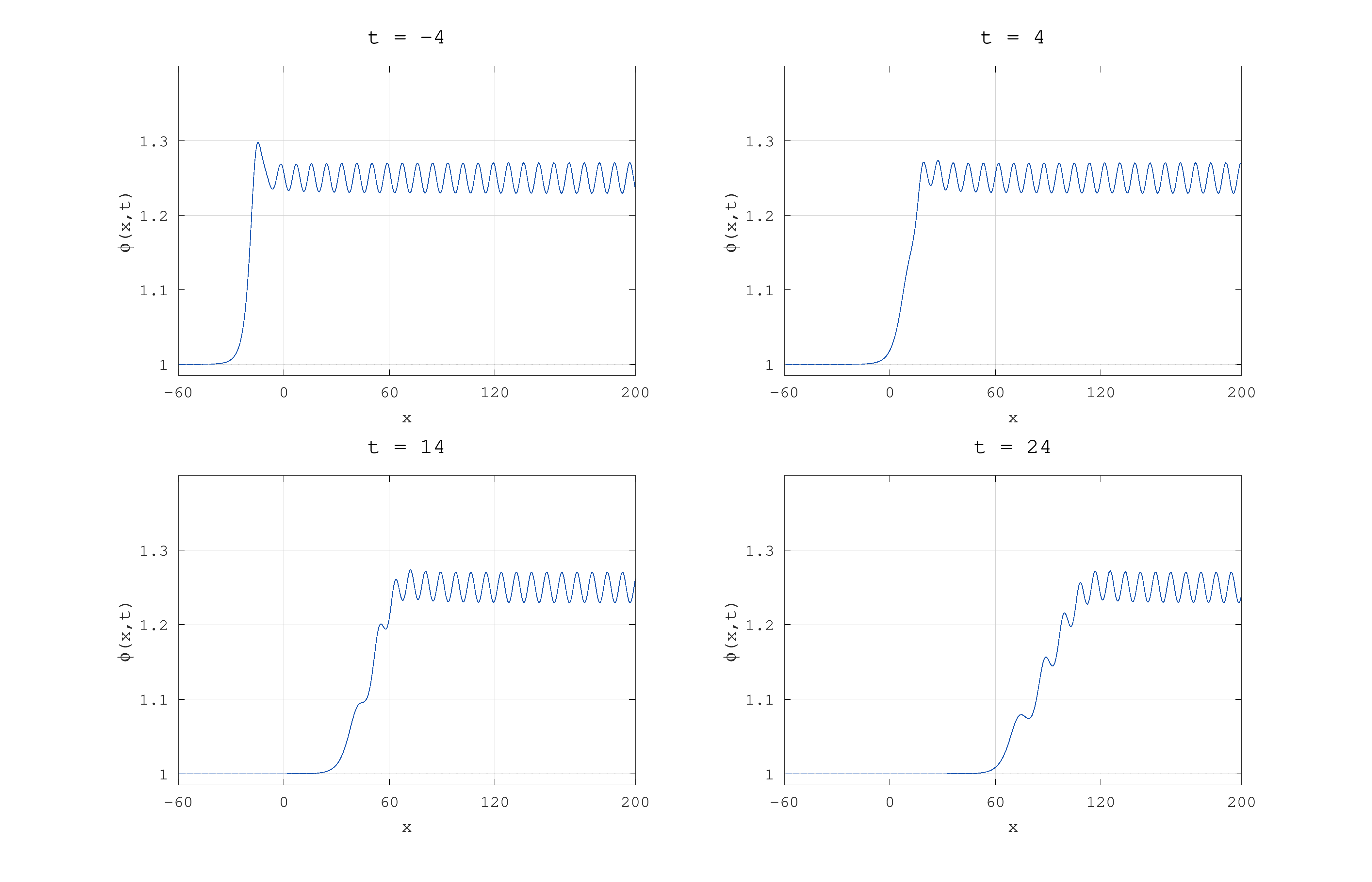}
	}
	\\
	\subfigure[soliton-soliton gas interaction]
	{
		\includegraphics[width=2.5in]{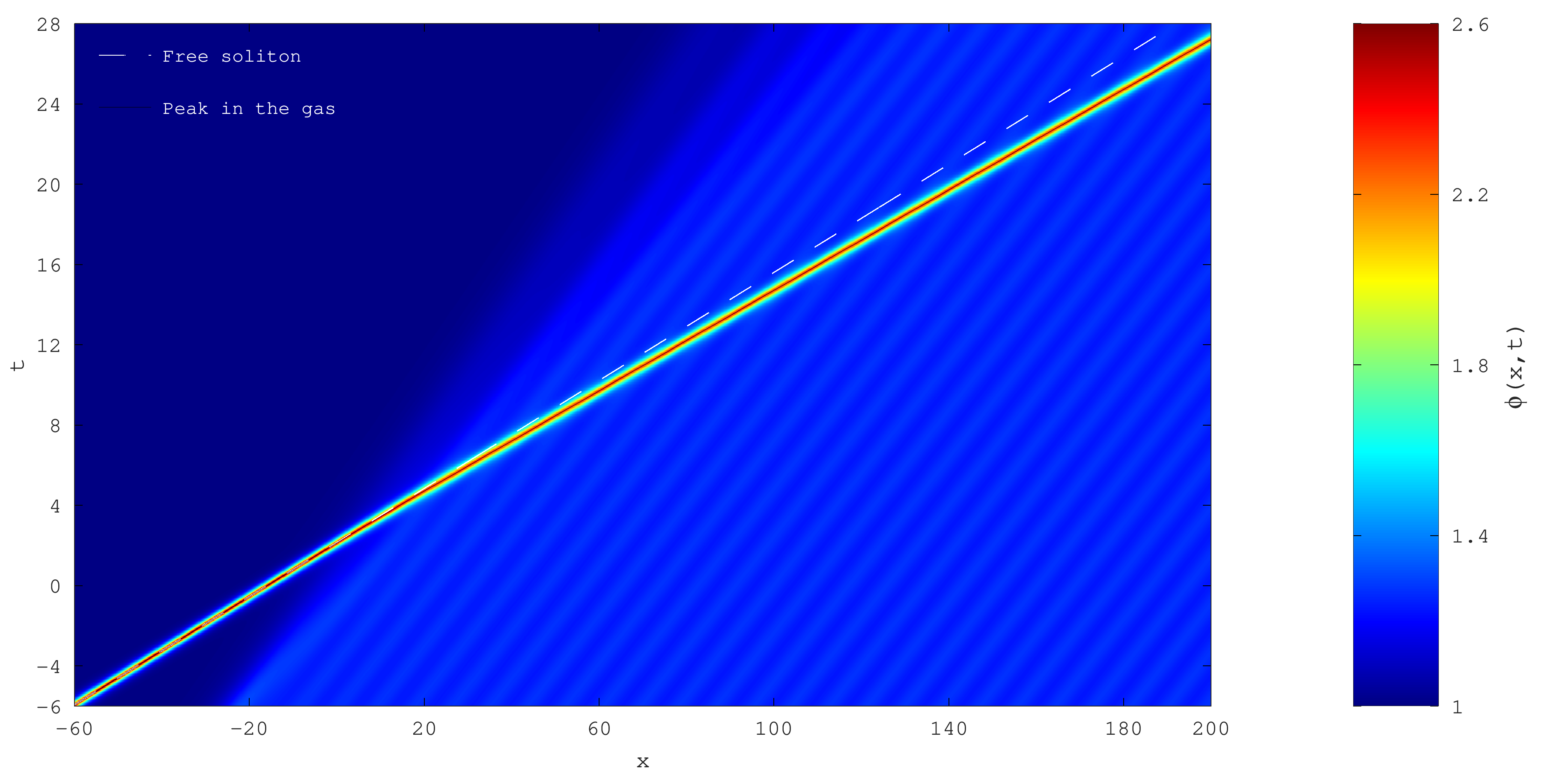}
	}
	\hspace*{3em}
	\subfigure[soliton-soliton gas interaction]
	{\
		\includegraphics[width=2.5in]{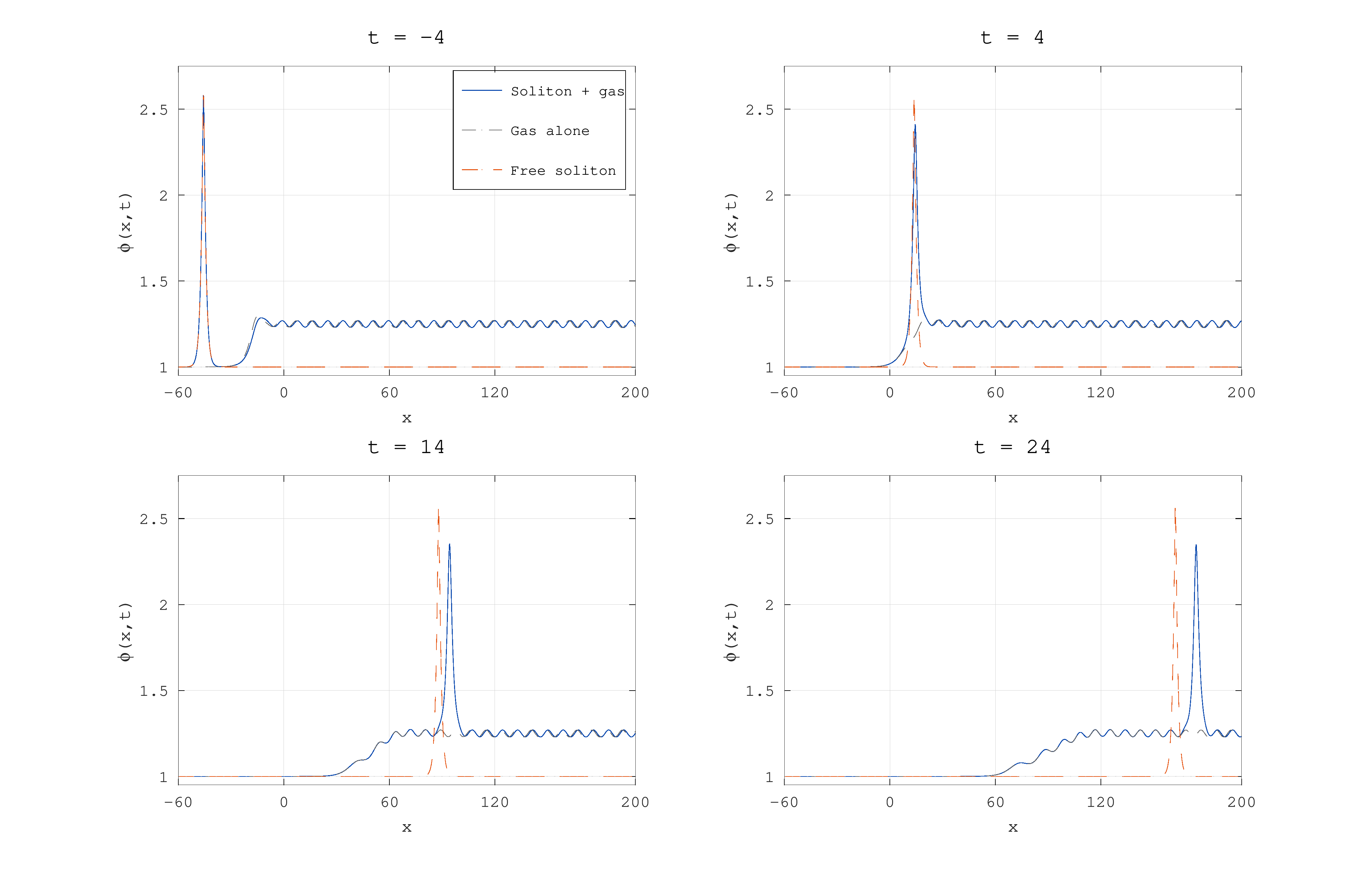}
	}
	\caption{(a)(b) Dense soliton gas to the mCH equation in Example \ref{ex:mch-separated} with $\Gamma=[0.1,0.3]$; (c)(d) interaction between soliton and soliton gas with $\Gamma=\{0.4\}\cup[0.1,0.3]$. }\label{mch-solitongas}
\end{figure}

\section{Proof of Thm. \ref{thm:main}}\label{sec:proof-main}

We use only standard trace-ideal facts.  If $\mathcal{K}(t)$ is differentiable in trace norm and $I+\mathcal{K}$ is invertible, then Jacobi's formula gives
\begin{equation}\label{eq:jacobi-fredholm}
\partial_t\log\det(I+\mathcal{K})=\Tr\bigl((I+\mathcal{K})^{-1}\mathcal{K}_t\bigr).
\end{equation}
Set
\[
H_\beta^j
=\{h:e^{-\beta u}h(u)\in H^j(\mathbb R_-)\},
\qquad j=0,1,
\]
and assume that, for some $\beta>0$, each $\mathcal K^{(n)}$
is trace class from $L^2(\mathbb R_-)$ into $H_\beta^1$.
The maps
\begin{align}
    Dh=h',\qquad
(\mathcal Vh)(u)=\int_{-\infty}^u h(v)\,\dd v,\label{maps}
\end{align}
are bounded inverses between $H_\beta^1$ and $H_\beta^0$.
Hence Sylvester's identity gives
\[
\begin{aligned}
\det(I+\mathcal D\mathcal K\mathcal V|_{H_\beta^0})
&=\det(I+\mathcal K\mathcal V\mathcal D|_{H_\beta^1})
 =\det(I+\mathcal K|_{H_\beta^1}),\\
\det(I+\mathcal V\mathcal K\mathcal D|_{H_\beta^1})
&=\det(I+\mathcal D\mathcal V\mathcal K|_{H_\beta^0})
 =\det(I+\mathcal K|_{H_\beta^0}).
\end{aligned}
\]
Applying the same identity to the natural inclusions
$H_\beta^1\subset H_\beta^0\subset L^2(\mathbb R_-)$
identifies both determinants with
$\det(I+\mathcal K|_{L^2(\mathbb R_-)})$.
\begin{proof}[Proof of Thm. \ref{thm:main}]
Put
\begin{equation}\label{eq:proof-notation}
 \mathcal{K}= \mathcal{K}^{(n)},\quad  \mathcal{K}^\pm= \mathcal{K}^{(n\pm1)},\quad
 U=(I+ \mathcal{K})^{-1},\quad U^\pm=(I+ \mathcal{K}^\pm)^{-1},\quad
 \mathscr C= \mathcal{K}_r+ \mathcal{K}_l^+.
\end{equation}
Write $K$ and $K^\pm$ for the respective integral kernels of $ \mathcal{K}$ and
$\mathcal{K}^\pm$.
For smoothing integral operators $\mathcal A$ and $\mathcal B$ with continuous kernels $A$
and $B$, respectively, write $[\mathcal A]=A(0,0)$ and $[\mathcal B]=B(0,0)$, and define
\begin{align}
(\mathcal A_l h)(u)
=\int_{-\infty}^0 \partial_u A(u,v)h(v)\,\dd v,
\qquad
(\mathcal A_r h)(u)
=\int_{-\infty}^0 \partial_v A(u,v)h(v)\,\dd v.\label{def-lr-deri}
\end{align}
Thus, the subscripts $l$ and $r$ denote differentiation
of the kernel in its first and second variables, respectively.
For products of integral operators,
\begin{align}
(\mathcal A\mathcal B)_l=\mathcal A_l\mathcal B,
\qquad
(\mathcal A\mathcal B)_r=\mathcal A\mathcal B_r.\label{brac-rule}
\end{align}
Higher and mixed derivatives are defined analogously. Integration by
parts on $\mathbb{R_-}$ gives the identities (also see in \cite{poppe1989})
\begin{equation}\label{eq:poppe-identities}
 [\mathcal A]=\Tr(\mathcal A_l+\mathcal A_r),\qquad [\mathcal A][\mathcal B]=[\mathcal A_r\mathcal B+\mathcal A\mathcal B_l].
\end{equation}
All boundary terms at $-\infty$ vanish by hypothesis.  The same identities,
together with $U=I-\mathcal{K}U=I-U\mathcal{K}$, imply, for every smoothing integral operator
$\mathcal T$,
\begin{align}
 [U\mathcal T U]&=\Tr(U \mathcal T_{x_1}-U\mathcal T U \mathcal{K}_{x_1}),\label{eq:resolvent-bracket-1}\\
 [U\mathcal T U^+]&=\Tr(U \mathcal T_r+\mathcal T_lU^+),\label{eq:resolvent-bracket-2}\\
 [U^+\mathcal T U]&=\Tr(\mathcal T_lU+U^+\mathcal T_r
 -\mathscr C U^+\mathcal T U).\label{eq:resolvent-bracket-3}
\end{align}
Here Eq. \eqref{eq:resolvent-bracket-1} assumes
$\mathcal T_{x_1}=\mathcal T_l+\mathcal T_r$, whereas
Eqs.\eqref{eq:resolvent-bracket-2}-\eqref{eq:resolvent-bracket-3} hold for every such
$\mathcal T$.  
These are identities between continuous endpoint kernels; no value of the distributional kernel of the identity operator is used.

Define
\begin{equation}\label{eq:w-def}
 w=(\log\tau_n)_{x_1}=[\mathcal{K}U]=[I-U],\qquad
 w^+=(\log\tau_{n+1})_{x_1}=[\mathcal{K}^+U^+]=[I-U^+].
\end{equation}
The primitives appearing below are normalized by
\begin{equation}\label{eq:x1-primitive}
 \partial_{x_1}^{-1}(w-w^+)_{x_j}
 =\partial_{x_j}\log\frac{\tau_n}{\tau_{n+1}}.
\end{equation}

\medskip
\noindent\emph{The first positive mKP equation.}
After division of Eq. \eqref{eq:mkp-first} by $\tau_n\tau_{n+1}$, its residual is
\begin{equation}\label{eq:first-normalized}
 w_{x_1}+w^+_{x_1}+(w-w^+)^2
 +\partial_{x_1}^{-1}(w-w^+)_{x_2}.
\end{equation}
Using  Eq. \eqref{eq:poppe-identities} and $\mathcal{K}_l+\mathcal{K}_r^+=0$ gives
\begin{equation}\label{eq:first-bracket}
 w_{x_1}+w^+_{x_1}+(w-w^+)^2
 =[U\mathscr C U^+]=\Tr(U\mathscr C_r+\mathscr C_lU^+).
\end{equation}
Moreover,
\begin{equation}\label{eq:first-x2}
 \partial_{x_1}^{-1}(w-w^+)_{x_2}
 =\Tr(\mathcal{K}_{x_2}U-\mathcal{K}^+_{x_2}U^+).
\end{equation}
Hence Eq. \eqref{eq:first-normalized} equals
\begin{align}
 \Tr\{U(\mathscr C_r+\mathcal{K}_{x_2})+U^+(\mathscr C_l-\mathcal{K}^+_{x_2})\}
 &=\Tr\{U\partial_l(\mathcal{K}_l+\mathcal{K}_r^+)+U^+\partial_r(\mathcal{K}_l+\mathcal{K}_r^+)\}=0,\label{first-eqn}
\end{align}
where  Eq. \eqref{eq:flow-x2} has been used. This proves  Eq. \eqref{eq:mkp-first}.

\medskip
\noindent\emph{The second positive mKP equation.}
A direct expansion of the Hirota operators gives
\begin{align}\label{eq:second-hirota-normalized}
 \frac{(D_{x_1}^3-3D_{x_1}D_{x_2}-4D_{x_3})
 \tau_n\cdot\tau_{n+1}}{\tau_n\tau_{n+1}}
 ={}&w_{x_1x_1}-w^+_{x_1x_1}
 +3(w-w^+)(w_{x_1}+w^+_{x_1})
 +(w-w^+)^3\notag\\
 &-3(w_{x_2}+w^+_{x_2})
 -3(w-w^+)\partial_{x_1}^{-1}(w-w^+)_{x_2}\notag\\
 &-4\partial_{x_1}^{-1}(w-w^+)_{x_3}.
\end{align}
Eliminating the $x_2$ primitive by the already proved first-flow identity
in  Eq. \eqref{eq:first-normalized}, the vanishing of
 Eq. \eqref{eq:second-hirota-normalized} is equivalent to $B_2=0$, where
\begin{align}\label{eq:B2-def-new}
 B_2={}&w_{x_1x_1}-w^+_{x_1x_1}
 -4\partial_{x_1}^{-1}(w-w^+)_{x_3}-3(w_{x_2}+w^+_{x_2})\notag\\
 &+(w-w^+)\{6(w_{x_1}+w^+_{x_1})+4(w-w^+)^2\}.
\end{align}

We now spell out the operator reduction of this residual.  Differentiating
 Eq. \eqref{eq:w-def} and using $U_{x_j}=-U\mathcal{K}_{x_j}U$ gives
\begin{align}
 w_{x_j}&=[U\mathcal{K}_{x_j}U],\label{eq:w-first-derivatives}\\
 w_{x_1x_1}
 &=\bigl[U\mathcal{K}_{x_1x_1}U-2U\mathcal{K}_{x_1}U\mathcal{K}_{x_1}U\bigr],\label{eq:w-second-x1}
\end{align}
with identical formulas for $w^+$.  In addition,
\begin{equation}\label{eq:primitive-trace-general}
 \partial_{x_1}^{-1}(w-w^+)_{x_j}
 =\Tr(\mathcal{K}_{x_j}U-\mathcal{K}^+_{x_j}U^+),
 \qquad j=2,3.
\end{equation}
Finally,  Eq. \eqref{eq:first-bracket} may be written as
\begin{equation}\label{eq:first-bracket-abbrev}
 w_{x_1}+w^+_{x_1}+(w-w^+)^2=[U\mathscr C U^+].
\end{equation}
Put $d=w-w^+=[U\mathcal{K}-U^+\mathcal{K}^+]$.  First,
\begin{align}
 6(w_{x_1}+w^+_{x_1})+4d^2
 &=2(w_{x_1}+w^+_{x_1})
   +4\{w_{x_1}+w^+_{x_1}+d^2\}\notag\\
 &=2\bigl([U\mathcal{K}_{x_1}U]+[U^+\mathcal{K}^+_{x_1}U^+]\bigr)
   +4[U\mathscr C U^+].
 \label{eq:B2-split}
\end{align}
The three products of brackets reduce term by term as follows:
\begin{align}
 4[U\mathcal{K}-U^+\mathcal{K}^+][U\mathscr C U^+]
 &=4\bigl[U(-\mathscr C_l+\mathcal{K}_{x_1}U\mathscr C)U^+
          +U^+\mathscr C_lU^+\bigr],
 \label{eq:B2-product-Q}\\
 2[U\mathcal{K}-U^+\mathcal{K}^+][U\mathcal{K}_{x_1}U]
 &=2\bigl[U(-\mathcal{K}_{x_1l}+\mathcal{K}_{x_1}U\mathcal{K}_{x_1})U
          +U^+\mathcal{K}_{x_1l}U\bigr],
 \label{eq:B2-product-F}\\
 2[U\mathcal{K}-U^+\mathcal{K}^+][U^+\mathcal{K}^+_{x_1}U^+]
 &=2\bigl[U(-\mathcal{K}^+_{x_1l}+\mathscr C U^+\mathcal{K}^+_{x_1})U^+\bigr]\notag\\
 &\quad
   +2\bigl[U^+(\mathcal{K}^+_{x_1l}-\mathcal{K}^+_{x_1}U^+\mathcal{K}^+_{x_1})U^+\bigr].
 \label{eq:B2-product-Fplus}
\end{align}
The factors $2$ in both terms on the right of
 Eq. \eqref{eq:B2-product-Fplus} are essential.  These formulas follow directly
from $[\mathcal A][\mathcal B]=[\mathcal{A}_r\mathcal{B}+\mathcal{A}\mathcal{B}_l]$, the two resolvent identities, and
$\mathcal{K}_l+\mathcal{K}_r^+=0$.

Together with
\begin{align}
 w_{x_1x_1}-w^+_{x_1x_1}
 &=[U\mathcal{K}_{x_1x_1}U]-2[U\mathcal{K}_{x_1}U\mathcal{K}_{x_1}U]\notag\\
 &\quad-[U^+\mathcal{K}^+_{x_1x_1}U^+]
       +2[U^+\mathcal{K}^+_{x_1}U^+\mathcal{K}^+_{x_1}U^+],
 \label{eq:B2-second-ledger}\\
 -3(w_{x_2}+w^+_{x_2})
 &=-3\bigl([U\mathcal{K}_{x_2}U]+[U^+\mathcal{K}^+_{x_2}U^+]\bigr),
 \label{eq:B2-x2-ledger}\\
 -4\partial_{x_1}^{-1}d_{x_3}
 &=-4\Tr(\mathcal{K}_{x_3}U-\mathcal{K}^+_{x_3}U^+),
 \label{eq:B2-x3-ledger}
\end{align}
 Eqs. \eqref{eq:B2-split}-\eqref{eq:B2-product-Fplus}
give
\begin{align}
 &w_{x_1x_1}-w^+_{x_1x_1}-3(w_{x_2}+w^+_{x_2})
 +d\{6(w_{x_1}+w^+_{x_1})+4d^2\}\notag\\
 &\qquad
 =\bigl[-4U\mathcal{K}_{x_2}U+2U^+\mathcal{K}^+_{x_2}U^+
        +2U^+\mathcal{K}_{x_1l}U+2UGU^+\bigr],
 \label{eq:B2-bracket-ledger}
\end{align}
where
\begin{equation}\label{eq:B2-G}
 G=-2\mathscr C_l+2\mathcal{K}_{x_1}U\mathscr C-\mathcal{K}^+_{x_1l}
   +\mathscr C U^+\mathcal{K}^+_{x_1}.
\end{equation}
The four remaining brackets are
\begin{align}
 -4[U\mathcal{K}_{x_2}U]
 &=\Tr\bigl(-4U\mathcal{K}_{x_1x_2}+4U\mathcal{K}_{x_2}U\mathcal{K}_{x_1}\bigr),
 \label{eq:B2-trace-1}\\
 2[U^+\mathcal{K}^+_{x_2}U^+]
 &=\Tr\bigl(2U^+\mathcal{K}^+_{x_1x_2}
       -2U^+\mathcal{K}^+_{x_2}U^+\mathcal{K}^+_{x_1}\bigr),
 \label{eq:B2-trace-2}\\
 2[U^+\mathcal{K}_{x_1l}U]
 &=\Tr\bigl(2\mathcal{K}_{x_1ll}U+2U^+\mathcal{K}_{x_1lr}
       -2\mathscr C U^+\mathcal{K}_{x_1l}U\bigr),
 \label{eq:B2-trace-3}\\
 2[UGU^+]
 &=\Tr\bigl(2G_lU^+ +2UG_r\bigr),
 \label{eq:B2-trace-4}
\end{align}
with
\begin{align*}
 G_l&=-2\mathscr C_{ll}+2\mathcal{K}_{x_1l}U\mathscr C-\mathcal{K}^+_{x_1ll}
      +\mathscr C_lU^+\mathcal{K}^+_{x_1},\\
 G_r&=-2\mathscr C_{lr}+2\mathcal{K}_{x_1}U\mathscr C_r-\mathcal{K}^+_{x_1lr}
      +\mathscr C U^+\mathcal{K}^+_{x_1r}.
\end{align*}
Substitution of
Eqs. \eqref{eq:B2-x3-ledger}-\eqref{eq:B2-bracket-ledger} and \eqref{eq:B2-trace-1}-\eqref{eq:B2-trace-4},
followed by cyclicity of the trace, gives
\begin{align}
 B_2={}&\Tr(LU)+\Tr(L^+U^+)\notag\\
 &+4\Tr\bigl(U\mathcal{K}_{x_1}U(\mathcal{K}_{x_2}+\mathscr C_r)\bigr)
 -2\Tr\bigl(U^+\mathcal{K}^+_{x_1}U^+(\mathcal{K}^+_{x_2}-\mathscr C_l)\bigr)\notag\\
 &+2\Tr\bigl(U\mathscr C U^+\partial_{x_1}(\mathcal{K}_l+\mathcal{K}_r^+)\bigr),
 \label{eq:B2-residual-new}
\end{align}
where
\begin{align}
 L={}&-4\mathcal{K}_{x_1x_2}+2\mathcal{K}_{x_1ll}-4\mathscr C_{lr}
      -2\mathcal{K}^+_{x_1lr}-4\mathcal{K}_{x_3},\label{eq:B2-L}\\
 L^+={}&2\mathcal{K}^+_{x_1x_2}+2\mathcal{K}_{x_1lr}-4\mathscr C_{ll}
      -2\mathcal{K}^+_{x_1ll}+4\mathcal{K}^+_{x_3}.
 \label{eq:B2-Lplus}
\end{align}
Set
\begin{equation}\label{eq:S-def}
 \mathscr A=\mathcal{K}_l+\mathcal{K}_r^+.
\end{equation}
The two nonlinear factors in  Eq. \eqref{eq:B2-residual-new} are immediately
\begin{align}
 \mathcal{K}_{x_2}+\mathscr C_r
 &=(\mathcal{K}_{ll}-\mathcal{K}_{rr})+(\mathcal{K}_{rr}+\mathcal{K}^+_{lr})
 =\mathcal{K}_{ll}+\mathcal{K}^+_{lr}=\mathscr A_l,\label{eq:B2-factor-left}\\
 \mathcal{K}^+_{x_2}-\mathscr C_l
 &=(\mathcal{K}^+_{ll}-\mathcal{K}^+_{rr})-(\mathcal{K}_{lr}+\mathcal{K}^+_{ll})
 =-(\mathcal{K}_{lr}+\mathcal{K}^+_{rr})=-\mathscr A_r.\label{eq:B2-factor-right}
\end{align}
It remains only to factor the two linear coefficients.  Using
\begin{equation*}
 \mathcal{K}_{x_1x_2}=\mathcal{K}_{lll}+\mathcal{K}_{llr}-\mathcal{K}_{lrr}-\mathcal{K}_{rrr},
 \qquad
 \mathcal{K}_{x_1ll}=\mathcal{K}_{lll}+\mathcal{K}_{llr},
 \qquad
 \mathscr C_{lr}=\mathcal{K}_{lrr}+\mathcal{K}^+_{llr},
\end{equation*}
and the corresponding formulas for $\mathcal{K}^+$, one obtains
\begin{align}
 L={}&-4(\mathcal{K}_{lll}+\mathcal{K}_{llr}-\mathcal{K}_{lrr}-\mathcal{K}_{rrr})
      +2(\mathcal{K}_{lll}+\mathcal{K}_{llr})\notag\\
 &-4(\mathcal{K}_{lrr}+\mathcal{K}^+_{llr})
      -2(\mathcal{K}^+_{llr}+\mathcal{K}^+_{lrr})
      -4(\mathcal{K}_{lll}+\mathcal{K}_{rrr})\notag\\
 ={}&-6(\mathcal{K}_{lll}+\mathcal{K}^+_{llr})-2(\mathcal{K}_{llr}+\mathcal{K}^+_{lrr})
 =-6\mathscr A_{ll}-2\mathscr A_{lr},
 \label{eq:B2-L-factorized}
\end{align}
whereas
\begin{align}
 L^+={}&2(\mathcal{K}^+_{lll}+\mathcal{K}^+_{llr}-\mathcal{K}^+_{lrr}-\mathcal{K}^+_{rrr})
      +2(\mathcal{K}_{llr}+\mathcal{K}_{lrr})\notag\\
 &-4(\mathcal{K}_{llr}+\mathcal{K}^+_{lll})
      -2(\mathcal{K}^+_{lll}+\mathcal{K}^+_{llr})
      +4(\mathcal{K}^+_{lll}+\mathcal{K}^+_{rrr})\notag\\
 ={}&2(\mathcal{K}_{lrr}+\mathcal{K}^+_{rrr})-2(\mathcal{K}_{llr}+\mathcal{K}^+_{lrr})
 =2\mathscr A_{rr}-2\mathscr A_{lr}.
 \label{eq:B2-Lplus-factorized}
\end{align}
Every summand in  Eq. \eqref{eq:B2-residual-new} therefore contains $\mathscr A$
or one of its left/right derivatives.  The adjacent-order condition in
 Eq. \eqref{eq:adjacent-order} is precisely $\mathscr A=0$, so $B_2=0$.  This
proves  Eq. \eqref{eq:mkp-second}.

\medskip
\noindent\emph{The first negative mKP equation.}
In addition to  Eq. \eqref{eq:x1-primitive}, normalize the primitive occurring
below by
\begin{equation}\label{eq:negative-x1-primitive}
 \partial_{x_1}^{-1}
 (w_{x_{-1}x_2}+w^+_{x_{-1}x_2})
 =\partial_{x_{-1}}\partial_{x_2}
 \log(\tau_n\tau_{n+1}).
\end{equation}
After expanding the Hirota operators and using the first positive mKP equation to eliminate the same $x_2$ primitive as above, the normalized residual of  Eq. \eqref{eq:mkp-negative} is
\begin{align}\label{eq:B3-def-new}
 B_3={}&w_{x_1x_{-1}}-w^+_{x_1x_{-1}}
 +2(w_{x_{-1}}+w^+_{x_{-1}})(w-w^+)\notag\\
 &+\partial_{x_1}^{-1}(w_{x_{-1}x_2}+w^+_{x_{-1}x_2})
 -4(w-w^+).
\end{align}
Here
\begin{equation}\label{eq:negative-w-derivative}
 w_{x_{-1}}=[U\mathcal{K}_{x_{-1}}U],
 \qquad
 w^+_{x_{-1}}=[U^+\mathcal{K}^+_{x_{-1}}U^+].
\end{equation}
In particular,
\begin{align}\label{eq:negative-mixed-derivative}
 w_{x_1x_{-1}}
 =\bigl[&U\mathcal{K}_{x_1x_{-1}}U
 -U\mathcal{K}_{x_1}U\mathcal{K}_{x_{-1}}U
 -U\mathcal{K}_{x_{-1}}U\mathcal{K}_{x_1}U\bigr],
\end{align}
with the analogous identity for $w^+$.  Put again $d=w-w^+$.  The product
term in  Eq. \eqref{eq:B3-def-new} is
\begin{align}
 2(w_{x_{-1}}+w^+_{x_{-1}})d
 ={}&2\bigl[U(-\mathcal{K}_{x_{-1}l}+\mathcal{K}_{x_1}U\mathcal{K}_{x_{-1}})U
              +U^+\mathcal{K}_{x_{-1}l}U\bigr]\notag\\
 &+2\bigl[U(-\mathcal{K}^+_{x_{-1}l}+\mathscr C U^+\mathcal{K}^+_{x_{-1}})U^+\bigr]\notag\\
 &+2\bigl[U^+(\mathcal{K}^+_{x_{-1}l}
              -\mathcal{K}^+_{x_1}U^+\mathcal{K}^+_{x_{-1}})U^+\bigr].
 \label{eq:B3-product-ledger}
\end{align}
Jacobi's formula gives the primitive term explicitly as
\begin{align}
 \partial_{x_1}^{-1}(w_{x_{-1}x_2}+w^+_{x_{-1}x_2})
 ={}&\Tr\bigl(\mathcal{K}_{x_{-1}x_2}U-\mathcal{K}_{x_2}U\mathcal{K}_{x_{-1}}U\bigr)\notag\\
 &+\Tr\bigl(\mathcal{K}^+_{x_{-1}x_2}U^+
            -\mathcal{K}^+_{x_2}U^+\mathcal{K}^+_{x_{-1}}U^+\bigr).
 \label{eq:B3-primitive-ledger}
\end{align}
Using $\mathcal{K}_{x_1}=\mathcal{K}_l+\mathcal{K}_r$, $\mathcal{K}_{x_2}=\mathcal{K}_{ll}-\mathcal{K}_{rr}$, and their
$x_{-1}$ derivatives in
 Eqs. \eqref{eq:negative-mixed-derivative}-\eqref{eq:B3-product-ledger}, we obtain
\begin{align}
 &w_{x_1x_{-1}}-w^+_{x_1x_{-1}}
 +2(w_{x_{-1}}+w^+_{x_{-1}})d\notag\\
 ={}&[U(\mathcal{K}_{x_{-1}r}-\mathcal{K}_{x_{-1}l})U]
 +[U(\mathcal{K}_{x_1}U\mathcal{K}_{x_{-1}}-\mathcal{K}_{x_{-1}}U\mathcal{K}_{x_1})U]\notag\\
 &-[U^+(\mathcal{K}^+_{x_{-1}r}-\mathcal{K}^+_{x_{-1}l})U^+]\notag\\
 &-[U^+(\mathcal{K}^+_{x_1}U^+\mathcal{K}^+_{x_{-1}}
        -\mathcal{K}^+_{x_{-1}}U^+\mathcal{K}^+_{x_1})U^+]\notag\\
 &+2[U(-\mathcal{K}^+_{x_{-1}l}+\mathscr C U^+\mathcal{K}^+_{x_{-1}})U^+]
 +2[U^+\mathcal{K}_{x_{-1}l}U].
 \label{eq:B3-bracket-ledger}
\end{align}
The three bracket reductions needed here are
\begin{align}
 \bigl[U(\mathcal{K}_{x_{-1}r}-\mathcal{K}_{x_{-1}l})U
 +U(\mathcal{K}_{x_1}U\mathcal{K}_{x_{-1}}-\mathcal{K}_{x_{-1}}U\mathcal{K}_{x_1})U\bigr]
 &=\Tr\bigl(-U\mathcal{K}_{x_{-1}x_2}+\mathcal{K}_{x_2}U\mathcal{K}_{x_{-1}}U\bigr),
 \label{eq:B3-trace-1}\\
 2[U(-\mathcal{K}^+_{x_{-1}l}+\mathscr C U^+\mathcal{K}^+_{x_{-1}})U^+]
 &=2\Tr\bigl((\mathscr C_lU^+\mathcal{K}^+_{x_{-1}}-\mathcal{K}^+_{x_{-1}ll})U^+\bigr)\notag\\
 &\quad
 +2\Tr\bigl(U(\mathscr C U^+\mathcal{K}^+_{x_{-1}r}-\mathcal{K}^+_{x_{-1}lr})\bigr),
 \label{eq:B3-trace-2}\\
 2[U^+\mathcal{K}_{x_{-1}l}U]
 &=2\Tr\bigl(\mathcal{K}_{x_{-1}ll}U+U^+\mathcal{K}_{x_{-1}lr}
             -\mathscr C U^+\mathcal{K}_{x_{-1}l}U\bigr).
 \label{eq:B3-trace-3}
\end{align}
Combining
 Eqs. \eqref{eq:B3-primitive-ledger}-\eqref{eq:B3-trace-3}
and cancelling the $\mathcal{K}_{x_{-1}x_2}$ terms gives
\begin{align}\label{eq:B3-residual-new}
 &w_{x_1x_{-1}}-w^+_{x_1x_{-1}}
 +2(w_{x_{-1}}+w^+_{x_{-1}})(w-w^+)\notag\\
 &\qquad
 +\partial_{x_1}^{-1}(w_{x_{-1}x_2}+w^+_{x_{-1}x_2})\notag\\
 &=4\Tr\bigl(
 -\mathcal{K}^+_{x_{-1}lr}U-\mathcal{K}^+_{x_{-1}rr}U^+
 +U\mathscr C U^+\mathcal{K}^+_{x_{-1}r}
 \bigr).
\end{align}
For completeness, we now derive each of the three negative-flow identities used to simplify the right-hand side.  From  Eq. \eqref{eq:flow-xm1},
\begin{equation}\label{eq:negative-flow-Fplus}
 \mathcal{K}^+_{x_{-1}}=\partial_l^{-1}\mathcal{K}^++\partial_r^{-1}\mathcal{K}^+.
\end{equation}
Differentiating once in $r$ and using the adjacent-order relation $\mathcal{K}_r^+=-\mathcal{K}_l$ gives
\begin{align}
 \mathcal{K}^+_{x_{-1}r}
 &=\partial_l^{-1}\mathcal{K}_r^++\mathcal{K}^+
 =-\partial_l^{-1}\mathcal{K}_l+\mathcal{K}^+
 =\mathcal{K}^+-\mathcal{K}.\label{eq:negative-id-r}
\end{align}
The vanishing normalization of the primitives at $-\infty$ is used in the last equality.  A further right derivative yields
\begin{equation}\label{eq:negative-id-rr}
 \mathcal{K}^+_{x_{-1}rr}=\mathcal{K}_r^+-\mathcal{K}_r.
\end{equation}
On the other hand, differentiating  Eq. \eqref{eq:negative-flow-Fplus} once in each kernel variable gives directly
\begin{equation}\label{eq:negative-id-lr}
 \mathcal{K}^+_{x_{-1}lr}=\mathcal{K}_r^++\mathcal{K}_l^+.
\end{equation}
Substitution of  Eq. \eqref{eq:negative-id-r}-\eqref{eq:negative-id-lr} into  Eq. \eqref{eq:B3-residual-new} gives
\begin{align}
 \mathcal N
 :=\frac14\,\text{right-hand side of \cref{eq:B3-residual-new}}
 ={}&\Tr\bigl(-(\mathcal{K}_l^++\mathcal{K}_r^+)U-(\mathcal{K}_r^+-\mathcal{K}_r)U^+\notag\\
 &\hspace{4.3em}+U\mathscr C U^+(\mathcal{K}^+-\mathcal{K})\bigr).
 \label{eq:B3-N}
\end{align}
The last term is reduced by the resolvent identity
\begin{equation}\label{eq:resolvent-difference}
 U^+(\mathcal{K}^+-\mathcal{K})U
 =U^+\bigl[(I+\mathcal{K}^+)-(I+\mathcal{K})\bigr]U
 =U-U^+.
\end{equation}
Cyclicity of the trace therefore gives
\begin{equation}\label{eq:B3-mixed-reduction}
 \Tr\bigl(U\mathscr C U^+(\mathcal{K}^+-\mathcal{K})\bigr)
 =\Tr\bigl(\mathscr C(U-U^+)\bigr).
\end{equation}
Using $\mathscr C=\mathcal{K}_r+\mathcal{K}_l^+$ in  Eq. \eqref{eq:B3-N} and then using
$\mathcal{K}_r^+=-\mathcal{K}_l$, we find
\begin{align}
 \mathcal N
 ={}&\Tr\bigl((\mathcal{K}_r-\mathcal{K}_r^+)U-(\mathcal{K}_r^++\mathcal{K}_l^+)U^+\bigr)\notag\\
 ={}&\Tr\bigl((\mathcal{K}_l+\mathcal{K}_r)U-(\mathcal{K}_l^++\mathcal{K}_r^+)U^+\bigr)\notag\\
 ={}&\Tr(\mathcal{K}_{x_1}U-\mathcal{K}^+_{x_1}U^+)
 =w-w^+.
 \label{eq:B3-final-trace}
\end{align}
Consequently, the first four terms in  Eq. \eqref{eq:B3-def-new} equal $4(w-w^+)$, which cancels the last term. Thus, $B_3=0$ and  Eq. \eqref{eq:mkp-negative} follow.

\medskip
\noindent\emph{The 2DTL equation.}
After division of Eq. \eqref{eq:2dtl-central} by $\tau_n^2$, it remains to prove
\begin{equation}\label{eq:toda-normalized}
 w_{x_{-1}}+\frac{\tau_{n+1}\tau_{n-1}}{\tau_n^2}-1=0,
 \qquad w_{x_{-1}}=[U\mathcal{K}_{x_{-1}}U].
\end{equation}
Let $\mathsf D$ and $\mathcal V$ follow the definition in Eq. \eqref{maps}.  Define
\begin{equation}\label{eq:toda-vectors}
 \begin{aligned}
  \operatorname{ev}_0(\varphi)&=\varphi(0),&
  \qquad \ell_0&=\operatorname{ev}_0\circ\mathcal V,\\
  \mathbf K_0(u)&=K(u,0),&
  \qquad \mathbf K_0^+(u)&=K^+(u,0),\\
  \mathbf v_+&=\mathcal V\mathbf K_0^+,&
  \qquad \mathbf v_-&=U\mathbf K_0.
 \end{aligned}
\end{equation}
Here, $\operatorname{ev}_0$ acts on $\mathcal H^1_\beta$ and $\ell_0$ on
$\mathcal H^0_\beta$; neither is being used as a bounded functional on $L^2(\R_-)$.

The adjacent-order relations at $n-1$ and $n$, followed by an integration
by parts, give the two operator identities
\begin{equation}\label{eq:toda-rank-one-operator}
 \mathsf D \mathcal{K}^-\mathcal V=\mathcal{K}-\mathbf K_0\otimes\ell_0,
 \qquad
 \mathcal V \mathcal{K}^+\mathsf D=\mathcal{K}+\mathbf v_+\otimes\operatorname{ev}_0.
\end{equation}
The first equality acts on $\mathcal H^0_\beta(\R_-)$ and the second on
$\mathcal H^1_\beta(\R_-)$, thus $\mathbf K_0\otimes\ell_0$ and
$\mathbf v_+\otimes\operatorname{ev}_0$ are bounded rank-one maps on the
respective levels.
Thus, we have
\begin{equation}\label{eq:toda-neighbour-ratios}
 \frac{\tau_{n-1}}{\tau_n}=1-\delta_-,
 \qquad
 \frac{\tau_{n+1}}{\tau_n}=1+\delta_+,
 \qquad
 \delta_-=\ell_0(\mathbf v_-),\quad
 \delta_+=\operatorname{ev}_0(U\mathbf v_+).
\end{equation}
Consequently,
\begin{equation}\label{eq:toda-product-ratio}
 \frac{\tau_{n+1}\tau_{n-1}}{\tau_n^2}-1
 =-\delta_-+\delta_+-\delta_-\delta_+.
\end{equation}

It remains to compute the logarithmic derivative in
 Eq. \eqref{eq:toda-normalized}.  For the smoothing operator
$\mathcal{K}_{x_{-1}}$, whose kernel is $K_{x_{-1}}$, the identity
$\mathcal{K}_{x_{-1}}U=\mathcal{K}_{x_{-1}}-\mathcal{K}_{x_{-1}}U\mathcal{K}$ and $K(\,\cdot\,,0)=\mathbf K_0$ show that the endpoint column
of $\mathcal{K}_{x_{-1}}U$ is $K_{x_{-1}}(\,\cdot\,,0)-\mathcal{K}_{x_{-1}}\mathbf v_-$.
Therefore, we have
\begin{equation}\label{eq:toda-endpoint-column}
 [U\mathcal{K}_{x_{-1}}U]
 =\operatorname{ev}_0U\bigl(K_{x_{-1}}(\,\cdot\,,0)-\mathcal{K}_{x_{-1}}\mathbf v_-\bigr).
\end{equation}
The negative flow and the adjacent-order relation give
\begin{equation}\label{eq:toda-endpoint-flow}
 K_{x_{-1}}(\,\cdot\,,0)=\mathcal V\mathbf K_0-\mathbf v_+.
\end{equation}
An integration by parts in the right primitive gives
\begin{align}
 \mathcal{K}_{x_{-1}}\mathbf v_-
 &=\mathcal V \mathcal{K}\mathbf v_--\mathcal{K}\mathcal V\mathbf v_-
   -\delta_-\mathbf v_+\notag\\
 &=\delta_-(\mathcal V\mathbf K_0-\mathbf v_+)
   -(\mathcal{K}-\mathcal{K}^-)\mathcal V\mathbf v_-,
 \label{eq:toda-F-p}
\end{align}
where the second equality uses
\begin{equation}\label{eq:toda-J-identities}
 \mathcal V\mathbf K_0=\mathcal V\mathbf v_-+\mathcal V \mathcal{K}\mathbf v_-,
 \qquad
 \mathcal{K}^-\mathcal V\mathbf v_-
 =\mathcal V \mathcal{K}\mathbf v_--\delta_-\mathcal V\mathbf K_0.
\end{equation}
Substitution of  Eq. \eqref{eq:toda-endpoint-flow}-\eqref{eq:toda-F-p} into
Eq. \eqref{eq:toda-endpoint-column} now gives the explicit cancellation
\begin{align}
 w_{x_{-1}}
 &=(1-\delta_-)\operatorname{ev}_0U
   (\mathcal V\mathbf K_0-\mathbf v_+)
   +\operatorname{ev}_0U(\mathcal{K}-\mathcal{K}^-)\mathcal V\mathbf v_-\notag\\
 &=-(1-\delta_-)\delta_+
   +\operatorname{ev}_0U\bigl(\mathcal V\mathbf K_0
     +\mathcal{K}\mathcal V\mathbf v_--\mathcal V \mathcal{K}\mathbf v_-\bigr)\notag\\
 &=-(1-\delta_-)\delta_+
   +\operatorname{ev}_0U(I+\mathcal{K})\mathcal V\mathbf v_-\notag\\
 &=\delta_--\delta_++\delta_-\delta_+.
 \label{eq:toda-log-derivative}
\end{align}
The identities in  Eqs. \eqref{eq:toda-product-ratio} and \eqref{eq:toda-log-derivative} prove
 Eq. \eqref{eq:toda-normalized}, and hence the 2DTL equation directly.
\end{proof}

\section{Non-isospectral reduction and deformed Bessel kernels}\label{sec:nonisospectral}

In this section, we establish the non-isospectral Bessel reductions presented in Sec. \ref{subsec:main-nonautonomous-results}. We begin in Sec. \ref{subsec:classical-bessel} and \ref{subsec:bessel-main} by providing the proofs of Prop. \ref{prop:cly-toda} and Thm. \ref{thm:bessel-reduction}, respectively. Subsequently, in Sec. \ref{subsec:bessel-inde}, we demonstrate how the finite-temperature deformed Bessel kernel naturally generates solutions to the reduced equations.

\subsection{Proof of Prop. \ref{prop:cly-toda}}\label{subsec:classical-bessel}
 \begin{proof}[Proof of Prop. \ref{prop:cly-toda}]
We consider the following kernel
\begin{align}
K_{\mathrm{Toda}}(u,v;x_{-1},x_1,\alpha) = \int_\gamma \frac{\dd q}{2\pi i}\ \int_{\hat\gamma} \frac{\dd p}{2\pi i} \frac{e^{(u+x_1) q + \frac{x_{-1}}{q} - (v+x_1) p - \frac{x_{-1}}{ p}}}{p - q} \left(\frac{p}{q}\right)^{\alpha} ,\quad u,v>0,\ x_{-1}<0,
\end{align}
acting on $L^2([0,1])$. A direct calculation verifies that $K_{\mathrm{Toda}}(u,v;x_{-1},x_1,\alpha)$ satisfies the differential relations \eqref{eq:flow-x1}, \eqref{eq:flow-xm1}, and \eqref{eq:adjacent-order}, where we treat the parameter $\alpha$ as the discrete lattice variable $n$. Although the kernel acts on $(0,1)$ instead of the half-line as in Thm. \ref{thm:main}, we can show that $\det\bigl(I-\mathcal{K}_{\mathrm{Toda}}(x_{-1},x_1,\alpha)\vert_{L^2(0,1)}\bigr)$ is a $\tau$-function for the 2DTL equation \eqref{eq:2dtl-central} using an argument analogous to that in Lem. \ref{lem:app-be-mkp}.

 On the other hand, applying the changes of variables  $\hat{q} =- \frac{q}{4x_{-1}}$ and $\hat{p} =- \frac{p}{4x_{-1}}$, the kernel transforms as
\begin{align*}
   K_{\mathrm{Toda}}(u,v;x_{-1},x_1,\alpha)&= -4x_{-1}\int_\gamma \frac{\dd \hat{q}}{2\pi i}\  \int_{\hat\gamma} \frac{\dd \hat{p}}{2\pi i} \frac{e^{-4x_{-1}(u+x_1) \hat{q} - \frac{1}{4\hat{q}} + 4x_{-1}(v+x_1)\hat{p} + \frac{1}{ 4\hat{p}}}}{\hat{p} - \hat{q}} \left(\frac{\hat{p}}{q}\right)^{\alpha}\\
   &=-4x_{-1}\left(\frac{u+x_1}{v+x_1}\right)^{\frac \alpha 2}K_{\mathrm{Be}}^{(\alpha)} (-4x_{-1}(u+x_1),-4x_{-1}(v+x_1)),
\end{align*}
Since the factor $\left(\frac{u+x_1}{v+x_1}\right)^{\frac \alpha 2}$ leaves the associated Fredholm determinant invariant, we have $G_\alpha(x_{-1},x_1)=\det(I-\mathcal{K}_{\mathrm{Toda}}(x_{-1},x_1,\alpha)|_{L^2([0,1])})=\det(I-\mathcal{K}_{\mathrm{Be}}^{(\alpha)}|_{L^2([s_1,s_2])})$ with $s_1=-4x_{-1}x_1,\ s_2=-4x_{-1}(1+x_1)$. Applying Jacobi's formula, we get
\begin{align*}
    \partial_{x_1}\log G_\alpha&=\partial_{x_1}\log\det(I-\mathcal{K}_{\mathrm{Be}}^{(\alpha)}|_{L^2([s_1,s_2])})=-4x_{-1}R(s_1,s_1)+4x_{-1}R(s_2,s_2),\\
    \partial_{x_{-1}}\log G_\alpha&=\partial_{x_{-1}}\log\det(I-\mathcal{K}_{\mathrm{Be}}^{(\alpha)}|_{L^2([s_1,s_2])})=-4x_1R(s_1,s_1)+4(1+x_1)R(s_2,s_2),
\end{align*}
where $R(u,v)$ is the kernel of the resolvent operator $\mathcal{R}=\mathcal{K}_{\mathrm{Be}}^{(\alpha)}(I-\mathcal{K}_{\mathrm{Be}}^{(\alpha)})^{-1}$. Note that when $x_1=0$ and $\alpha>0$, the above formulas reduce to 
\begin{align}
    ( \partial_{x_1}-x_{-1}\partial_{x_{-1}})\log G_\alpha=-4x_{-1}R(0,0)=0.
\end{align}
Thus, the 2DTL equation \eqref{eq:2dtl-central} reduces to \eqref{eqn:cyl-toda}. This finishes the proof of the proposition.

 \end{proof}
 
The above proof provides an alternative derivation of the well-known connection between the classical Bessel kernel and the Toda lattice.

\subsection{Proof of Thm. \ref{thm:bessel-reduction}}\label{subsec:bessel-main}

\begin{proof}[Proof of Thm. \ref{thm:bessel-reduction}]
    Let $\tau_n(x_1,x_2,x_{-1})$ satisfy  Eqs. \eqref{eq:mkp-first} and \eqref{eq:2dtl-central}. We define the scaled function
\[
\widehat\tau_n(x,t)
=x^{(4(\alpha+n)^2-1)/8}
\tau_n\!\left(\frac{x^2}{4},0,t\right),
\qquad n=-1,0,1.
\]
Applying the reduction constraint \eqref{noniso-reduction} to eliminate the $x_2$-derivatives, and using the fact $\partial_{x_1}=2x^{-1}\partial_x$, the first positive member of the mKP hierarchy \eqref{eq:mkp-first} at $n=-1, 0,$ becomes
\[
D_x^2\widehat\tau_{-1}\cdot\widehat\tau_0=0,
\qquad
D_x^2\widehat\tau_0\cdot\widehat\tau_1=0.
\]
Similarly, evaluating the 2DTL equation \eqref{eq:2dtl-central} at $n=0$ yields
\[
(D_xD_t-x)\widehat\tau_0\cdot\widehat\tau_0
=-\widehat\tau_1\widehat\tau_{-1}.
\]
Denote
\[
\phi=-\partial_x\log\widehat\tau_0+\frac{xt}{2},
\qquad
\psi_\pm=\frac{\widehat\tau_{\pm1}}{\widehat\tau_0}.
\]
Then, the above bilinear identities give us
\[
(\psi_\pm)_{xx}=(2\phi_x-t)\psi_\pm,
\qquad
\phi_t=\frac12\psi_+\psi_-.
\]
Differentiating the above formula and combining them together, we obtain
\begin{equation}
(\psi_+\psi_-)_{xxx}
=4(2\phi_x-t)(\psi_+\psi_-)_x
 +4\phi_{xx}\psi_+\psi_-,
\end{equation}
which is equivalent to Eq. \eqref{n-nkdv1}. This completes the proof of the theorem.
\end{proof}

\subsection{Reduction to the deformed Bessel determinant}\label{subsec:bessel-inde}
Let the weight $\sigma$ satisfy \cite[Assump.~1.1]{ruzza2025}. To rigorously justify the required kernel operations, we first assume $\alpha>5/2$. On the space $L^2(0,x_1)$, we introduce the operator $\mathcal K^{(\mu)}_0$ with kernel
\[
K^{(\mu)}_0(u,v;t)
=\left(\frac vu\right)^{\mu/2}
\int_0^\infty \sigma(t+\lambda)
J_\mu(2\sqrt{u\lambda})
J_\mu(2\sqrt{v\lambda})\,d\lambda,
\qquad \mu=\alpha+n.
\]
To introduce the auxiliary variable $x_2$, we extend the kernel to
\begin{align}
 K^{(\mu)}(u,v;x_2,t)
 =K^{(\mu)}_0(u,v;t)
  +x_2(\partial_u^2-\partial_v^2)K^{(\mu)}_0(u,v;t),\label{eq:app-be-extension}
\end{align}
and define
\begin{equation} \label{eqn:taun-bessel-def}
  \tau_n(x_1,x_2,t)
 =\det\bigl(I-\mathcal K^{(\alpha+n)}(x_2,t)\bigr)_{L^2(0,x_1)},
 \qquad n=-1,0,1.   
\end{equation}
For our purposes, we only require the behavior of this extended kernel and its first $x_2$-derivative evaluated at $x_2=0$. 
The following lemma ensures that this extension is well-defined at the operator level.
\begin{lemma}\label{lem:app-be-analytic}

For each $x_1 > 0$ and $n \in \{-1,0,1\}$, the integral operator $\mathcal K^{(\alpha+n)}(x_2,t)$ acting on $L^2(0,x_1)$ with kernel $K^{(\alpha+n)}(u,v;x_2,t)$ given by \eqref{eq:app-be-extension} is trace class, with trace norm bounded locally uniformly in $(x_1,x_2,t)$. Moreover, we have
\begin{equation}\label{eq:app-be-identification}
 \tau_n(x_1,0,t)
 =Q^{\mathrm{Be}}_{\alpha+n,\sigma}(2\sqrt{x_1},t),
 \qquad n=-1,0,1,
\end{equation}
where $Q_{\alpha+n,\sigma}^{\mathrm{Be}}$ is defined in \eqref{eq:Bessel-deformed-det}.
\end{lemma}

\begin{proof}
We write
\begin{align}
 \Phi^{(\mu)}(u,\lambda)=\left(\frac\lambda u\right)^{\mu/2}
 J_\mu(2\sqrt{u\lambda}),\qquad
 \Psi^{(\mu)}(v,\lambda)=\left(\frac v\lambda\right)^{\mu/2}
 J_\mu(2\sqrt{v\lambda}).
\end{align}
The recurrence relations for Bessel functions (cf. \cite[Eq. (10.6.6)]{DLMF}) yield
\begin{equation}\label{eq:app-be-wave-derivatives}
 \partial_u\Phi^{(\mu)}=-\Phi^{(\mu+1)},\qquad
 \partial_v\Psi^{(\mu)}=\Psi^{(\mu-1)},\qquad
 \mathcal K^{(\mu)}_{l}+\mathcal K^{(\mu+1)}_{r}=0.
\end{equation}
The last identity is used for $n=-1,0$. The operators $\mathcal K_l$ and $\mathcal K_r$ are defined as in Eq. \eqref{def-lr-deri}, except that they act here on \(L^2(0,x_1)\).
Since $\mu=\alpha+n>3/2$, we have the trace-norm estimate for $\mathcal K^{(\mu)}$: 
\begin{equation}\label{eq:app-be-trace-bound}
 \|\mathcal K^{(\mu)}_{l^i r^j}\|_1
 \leq\int_0^\infty |\sigma(t+\lambda)|
 \|\partial_u^i\Phi^{(\mu)}(\cdot,\lambda)\|_2
 \|\partial_v^j\Psi^{(\mu)}(\cdot,\lambda)\|_2\,\dd\lambda<\infty,
 \qquad 0\leq i,j\leq2.
\end{equation}
Indeed, the product of the norms is $O(\lambda^{\mu+i})$ at zero
and grows at most polynomially at infinity, where the decay of
$\sigma$ ensures integrability. These estimates are locally uniform
in $(x_1,t)$. Thus $\mathcal K^{(\mu)}$ and its auxiliary $x_2$ extension
are trace class. Furthermore, these bounds also give us
\begin{equation}\label{esti-low-bound}
\partial_u^i\partial_v^j K^{(\mu)}(u,v;t)=O(v^{\mu-j}),
\qquad v\rightarrow0,\ 0\leq i,j\leq2, 
\end{equation}
uniformly for $u\in [0,x_1]$. These estimates justify the application of Sylvester's determinant identity. Since the gauge factors leave the Fredholm determinant invariant, we obtain
\begin{equation}
\tau_n(x_1,0,t)
=\det(I-\mathcal K^{(\alpha+n)})_{L^2(0,x_1)}
=Q_{\alpha+n,\sigma}^{\rm Be}(2\sqrt{x_1},t).
\end{equation} 
This finishes the proof of the lemma.
\end{proof}

With the analytic properties of the extended determinant established, we next show that $\tau_n$ satisfies the first positive member of the mKP hierarchy, the 2DTL equation, and the non-isospectral reduction on $x_2=0$. 

\begin{lemma}\label{lem:app-be-mkp}
Let $\tau_n(x_1,x_2,t)$ be defined in Eq. \eqref{eqn:taun-bessel-def}. Then, evaluated at $x_2=0$, the following relations hold:
\begin{align}
 &(D_{x_1}^2+D_{x_2})\tau_n\cdot\tau_{n+1}=0,
 \qquad n=-1,0,
 \label{eq:app-be-mkp}\\
 &\left(\frac12D_{x_1}D_t-1\right)
 \tau_0\cdot\tau_0
 =-\tau_1\tau_{-1},\label{eq:app-be-toda}\\
 &\left(x_1\partial_{x_2}+(\alpha+n)\partial_{x_1}\right)
 \tau_n=0,\qquad n=-1,0,1.
 \label{eq:app-be-constraint}
\end{align}
\end{lemma}

\begin{proof}
Throughout this proof, all expressions are evaluated at $x_2=0$, and we drop this evaluation from the notation for brevity. 

For an integral operator $\mathcal B$ acting on $L^2(0,x_1)$ with a continuous kernel $B(u,v)$ satisfying the decay condition \eqref{esti-low-bound} near the origin, we define the bracket notation $[\mathcal B]=B(x_1,x_1)$. Similar to Eq. \eqref{eq:poppe-identities}, we have the identities
\begin{align}
    [\mathcal B]=\operatorname{Tr}(\mathcal B_l+\mathcal B_r),
    \qquad
    [\mathcal B][\mathcal C]
    =[\mathcal B_r\mathcal C+\mathcal B\mathcal C_l].\label{brac-identity-bessel}
\end{align}
These identities follow easily from the decay estimate Eq. \eqref{esti-low-bound}, which ensures the boundary terms at the origin vanish. Indeed, for the first identity, we have
\begin{align}
    \operatorname{Tr}(\mathcal B_l+\mathcal B_r)=\int_0^{x_1}(\partial_u B(u,v)+\partial_v B(u,v))\big|_{v=u} \dd u=B(x_1,x_1)-B(0,0)=[\mathcal B].
\end{align}
For the second identity, we get
\begin{align}
    [\mathcal B_r\mathcal C+\mathcal B\mathcal C_l]&=\int_0^{x_1} \Big( \partial_sB(x_1,s)C(s,x_1)+B(x_1,s)\partial_s C(s,x_1) \Big) \dd s =\int_0^{x_1}\partial_s\Big( B(x_1,s)C(s,x_1) \Big) \dd s\notag\\
    &=B(x_1,x_1)C(x_1,x_1)-B(x_1,0)C(0,x_1) =[\mathcal B][\mathcal C].
\end{align}
Next, let
\begin{align*}
    U^{(\mu)}=(I-\mathcal K^{(\mu)})^{-1},
    \qquad
    \mathcal R^{(\mu)}=\mathcal K^{(\mu)} U^{(\mu)},
    \qquad
    w=(\log\tau_n)_{x_1}.
\end{align*}
where $R^{(\mu)}(u,v)$ denotes the kernel of the resolvent $\mathcal R^{(\mu)}$. Jacobi's determinant formula gives
\begin{align}
    w=-[\mathcal R^{(\mu)}] = -R^{(\mu)}(x_1,x_1).
\end{align}
Moreover, we have
\begin{align}
    \partial_{x_1}R^{(\mu)}(u,v)
    =R^{(\mu)}(u,x_1)R^{(\mu)}(x_1,v),
\end{align}
and hence
\begin{align}
    w_{x_1}
    =-[\mathcal R_{l}^{(\mu)}+\mathcal R_{r}^{(\mu)}]
     -[\mathcal R^{(\mu)}]^2
    =-[U^{(\mu)}(\mathcal K_{l}^{(\mu)}+\mathcal K_{r}^{(\mu)})U^{(\mu)}].
\end{align}
Using the identities \eqref{brac-identity-bessel} alongside the relations
\begin{align}
\mathcal K^{(\mu)}_l+\mathcal K^{(\mu+1)}_r=0,
\qquad
\mathcal K^{(\mu)}_{x_2}
=\mathcal K^{(\mu)}_{ll}-\mathcal K^{(\mu)}_{rr},
\end{align}
we can repeat the calculation in Eqs. \eqref{eq:first-normalized}-\eqref{first-eqn}, which establishes \eqref{eq:app-be-mkp}.

To prove Eq. \eqref{eq:app-be-toda}, we define
\begin{align}
(\mathcal V_1 h)(u)=\int_u^{x_1}h(v)\,dv.
\end{align}
The recurrence
$\mathcal K^{(\mu)}_{l}+\mathcal K^{(\mu+1)}_{r}=0$
and the boundary condition $K^{(\mu)}(u,0)=0$ give
\begin{align*}
\mathcal K^{(\alpha+1)}
&=-\mathcal K^{(\alpha)}_{l}\mathcal V_1,
&
\mathcal V_1\mathcal K_{l}^{(\alpha)}
&=\mathbf 1\otimes K^{(\alpha)}(x_1,\cdot)-\mathcal K^{(\alpha)},
\\
\mathcal K^{(\alpha)}
&=\mathcal K_{r}^{(\alpha)}\mathcal V_1,
&
\mathcal V_1\mathcal K_{r}^{(\alpha)}
&=\mathcal K^{(\alpha-1)}
  -\mathbf 1\otimes K^{(\alpha-1)}(x_1,\cdot).
\end{align*}
Here $\mathbf 1$ is the constant function, and
$\mathbf 1\otimes K^{(\mu)}(x_1,\cdot)$ denotes the rank-one
operator with kernel $K^{(\mu)}(x_1,v)$.

Applying Sylvester's determinant identity and the rank-one determinant formula, we have
\begin{align}
\frac{\tau_1}{\tau_0}
=1+[U^{(\alpha)}\mathcal K^{(\alpha)}\mathcal V_1],
\qquad
\frac{\tau_{-1}}{\tau_0}
=1-[\mathcal K^{(\alpha-1)}\mathcal V_1 U^{(\alpha)}].
\end{align}
Next, changing variables to $\rho=t+\lambda$ in the defining
integral for $K^{(\alpha)}$ and using the Bessel recurrences gives
\begin{align}
K_{t}^{(\alpha)}(u,v)
=\int_0^v
 \bigl(K^{(\alpha)}(u,s)-K^{(\alpha-1)}(u,s)\bigr)\,ds,
\end{align}
or equivalently,
\[
\mathcal K_{t}^{(\alpha)}
=(\mathcal K^{(\alpha)}-\mathcal K^{(\alpha-1)})\mathcal V_1.
\]
This differentiation holds in trace norm by the previous rank-one estimates, i.e., the Bessel integrand and its
$\lambda$-derivative are respectively
$O(\lambda^\alpha)$ and $O(\lambda^{\alpha-1})$ at zero, and the rapid decay of $\sigma$ ensures integrability of the tail.

Since
\begin{align}
(U^{(\alpha)}\mathcal K^{(\alpha)} \mathcal V_1)_r=U^{(\alpha)}\mathcal K^{(\alpha)},
\qquad
(\mathcal K^{(\alpha-1)}\mathcal V_1 U^{(\alpha)})_l=-\mathcal K^{(\alpha)} U^{(\alpha)},
\end{align}
the bracket product identity gives
\[
\begin{aligned}
&[U^{(\alpha)}\mathcal K^{(\alpha)} \mathcal V_1]\,
 [\mathcal K^{(\alpha-1)} \mathcal V_1 U^{(\alpha)}]
=[U^{(\alpha)}\mathcal K^{(\alpha)} \mathcal K^{(\alpha-1)}\mathcal V_1 U^{(\alpha)}
  -U^{(\alpha)}\mathcal K^{(\alpha)} \mathcal V_1\mathcal K^{(\alpha)} U^{(\alpha)}].
\end{aligned}
\]
Finally, the differentiation formula gives
\[
\partial_{x_1}\log\tau_0=-[\mathcal K^{(\alpha)} U^{(\alpha)}].
\]
Differentiating in $t$ and using
$U^{(\alpha)}-I=U^{(\alpha)}\mathcal K^{(\alpha)}=\mathcal K^{(\alpha)} U^{(\alpha)}$, we obtain
\[
\begin{aligned}
\partial_{x_1}\partial_t\log\tau_0
&=-[U^{(\alpha)}\mathcal K_{t}^{(\alpha)} U^{(\alpha)}]
\\
&=-[U^{(\alpha)}(\mathcal K^{(\alpha)} -\mathcal K^{(\alpha-1)} )\mathcal V_1 U^{(\alpha)}]
\\
&=-[U^{(\alpha)}\mathcal K^{(\alpha)} \mathcal V_1]
  +[\mathcal K^{(\alpha-1)} \mathcal V_1 U^{(\alpha)}]
  +[U^{(\alpha)}\mathcal K^{(\alpha)} \mathcal V_1]\,
   [\mathcal K^{(\alpha-1)} \mathcal V_1 U^{(\alpha)}]
\\
&=1-\frac{\tau_1\tau_{-1}}{\tau_0^2}.
\end{aligned}
\]
This is equivalent to the bilinear form \eqref{eq:app-be-toda}.

For the constraint \eqref{eq:app-be-constraint}, fix $n\in\{-1,0,1\}$ and let $\mu=\alpha+n$.
Jacobi's formula and integration by parts give
\begin{equation}\label{eq:app-be-constraint-trace}
 \partial_{x_2}\log \tau_n
 =-\operatorname{Tr}(\mathcal R_{ll}^{(\mu)} -\mathcal R_{rr}^{(\mu)})
 =-\left[(\partial_u-\partial_v)R^{(\mu)}(u,v)\big|_{v=u}\right]_{u=0}^{u=x_1}.
\end{equation}
Since 
\begin{align}
    K^{(\mu)}(u,v)=\left(\frac vu\right)^{\mu/2}G(u,v),\quad \textrm{with } G(u,v)=G(v,u),
\end{align}
the resolvent inherits this symmetry, so we can write $R^{(\mu)}(u,v)=(v/u)^{\mu/2}\widetilde{R}^{(\mu)}(u,v)$ with $\widetilde{R}^{(\mu)}(u,v)=\widetilde{R}^{(\mu)}(v,u)$. Hence,
\begin{align}
 (\partial_u-\partial_v)R^{(\mu)}(u,v)\big|_{v=u}
 &=\left(-\frac{\mu}{2u} R^{(\mu)}(u,v)+\left(\frac{v}{u}\right)^{\frac \mu2}\widetilde{R}_u^{(\mu)}(u,v)-\frac{\mu}{2v} R^{(\mu)}(u,v)-\left(\frac{v}{u}\right)^{\frac \mu2}\widetilde{R}_v^{(\mu)}(u,v)\right)\Big|_{v=u}\notag\\
 &=-\frac{\mu}{u} R^{(\mu)}(u,u).
\end{align}
The boundary term at zero vanishes since $R^{(\mu)}(u,u)=O(u^\mu)$. Thus, Eq. \eqref{eq:app-be-constraint-trace} simplifies to $\partial_{x_2}\log\tau_n
=-\mu\partial_{x_1}\log\tau_n/x_1$, which proves Eq. \eqref{eq:app-be-constraint}.
\end{proof}

Eqs. \eqref{eq:app-be-mkp} and \eqref{eq:app-be-toda} correspond exactly to the first member of the mKP hierarchy and the 2DTL equation, respectively, while \eqref{eq:app-be-constraint} is precisely the non-isospectral constraint required by Thm. \ref{thm:bessel-reduction}. Therefore, applying this reduction with $x_1=x^2/4$  shows that
\begin{equation}
 \phi(x,t)
 =-\partial_x\log Q_{\alpha,\sigma}^{\mathrm{Be}}(x,t)
   +\frac{xt}{2}-\frac{4\alpha^2-1}{8x},
\end{equation}
satisfies Eq. \eqref{n-nkdv1} wherever
$Q_{\alpha,\sigma}^{\mathrm{Be}}(x,t)\ne0$.

\section*{Acknownledgement}
This work was supported by grants from the Research Grants Council of the Hong Kong Special Administrative Region, China (Project No. CityU 11306723, CityU 11301924 and CityU 11301225).

\appendix
\section{Fredholm determinant solution to the fNLS equation}\label{pf-fnls}
In this appendix, we will present the Fredholm determinant solution of the fNLS equation
\begin{align}
    \mathrm i\phi_t-\phi_{xx}-2|\phi|^2\phi=0, \label{eq:fnls}
\end{align}
 as well as the construction of the breather gas  solution.  We first establish a Fredholm determinant lemma containing the bilinear equations required for the reduction to the fNLS equation. We then carry out the reduction and obtain a Fredholm
determinant representation of the solution to the fNLS equation, followed by an explicit example producing the
breather gas solution.

\begin{lemma}[An alternative Fredholm determinant representation of the mKP-2DTL equations]
\label{lem:fnls-mother}
Let $P,Q$ be bounded invertible operators on a complex Hilbert space $H$
and let $\mathcal K_0$ be trace class, with
\begin{equation}
 \operatorname{rank}(P\mathcal K_0-\mathcal K_0Q)\le1.
 \label{eq:fnls-mother-rank}
\end{equation}
For $\mathbf x=(x_{-1},x_1,x_2)\in\mathbb C^3$, put
\begin{equation}
 \theta(z)=x_{-1}z^{-1}+x_1z+x_2z^2,
 \label{eq:fnls-mother-phase}
\end{equation}
and define
\begin{equation}
 \mathcal K^{(n)}=P^n e^{\theta(P)}\mathcal K_0e^{-\theta(Q)}Q^{-n},\qquad
 \tau_n=\det(I+\mathcal K^{(n)}),\quad n\in\mathbb Z.
 \label{eq:fnls-mother-tau}
\end{equation}
Then $\tau_n$ satisfies the first mKP equation \eqref{eq:mkp-first} and the 2DTL equation \eqref{eq:2dtl-central}.
\end{lemma}

\begin{proof}
The operators $\mathcal K^{(n)}$ and their parameter derivatives are trace class, since
the factors in Eq. \eqref{eq:fnls-mother-tau} are bounded and entire
in the parameters. For fixed $n$, write $\mathcal K=\mathcal K^{(n)}$,
and factor
\[
 P\mathcal K-\mathcal KQ=a\otimes b,
 \qquad (a\otimes b)h=a\,b(h),
\]
where $a$ is a vector and $b$ is a continuous linear functional.
The factors may be chosen so that
\[
 a_{x_j}=P^ja,\qquad b_{x_j}=-bQ^j,
 \qquad j=-1,1,2.
\]
The definition gives
\begin{align*}
 \mathcal K^{(n+1)}-\mathcal K&=a\otimes bQ^{-1},&
 \mathcal K^{(n-1)}-\mathcal K&=-P^{-1}a\otimes b,\\
 \mathcal K_{x_1}&=a\otimes b,&
 \mathcal K_{x_{-1}}&=-P^{-1}a\otimes bQ^{-1},\\
 \mathcal K_{x_2}&=Pa\otimes b+a\otimes bQ.
\end{align*}
For $\tau_n\ne0$, we introduce
\[
 U=(I+\mathcal K)^{-1},\qquad
 h=\frac{\tau_{n+1}}{\tau_n},\qquad
 w=(\log\tau_n)_{x_1}.
\]
The rank-one determinant formula and Jacobi's formula give
\begin{equation}
 h=1+bQ^{-1}Ua,\qquad
 \frac{\tau_{n-1}}{\tau_n}=1-bUP^{-1}a,\qquad
 w=bUa.
 \label{eq:fnls-mother-ratios}
\end{equation}
Using $U_{x_j}=-U\mathcal K_{x_j}U$, we obtain
\begin{align*}
 w_{x_{-1}}
 &=-bQ^{-1}Ua+bUP^{-1}a
   +(bUP^{-1}a)(bQ^{-1}Ua)\\
 &=1-\frac{\tau_{n+1}\tau_{n-1}}{\tau_n^2},
\end{align*}
which proves Eq. \eqref{eq:2dtl-central}.

Similarly, to establish the first mKP equation, we compute the following derivatives:
\begin{align*}
 w_{x_1}&=bUPa-bQUa-w^2,\\
 h_{x_1}&=bQ^{-1}UPa-hw,\\
 h_{x_2}&=bQ^{-1}UP^2a-w\,bQ^{-1}UPa-h\,bQUa,\\
 h_{x_1x_1}&=bQ^{-1}UP^2a-h\,bUPa-w h_{x_1}-h w_{x_1}.
\end{align*}
Consequently, we obtain
\[
 h_{x_2}-h_{x_1x_1}=2w_{x_1}h,
\]
which is equivalent to Eq. \eqref{eq:mkp-first} after division by
$\tau_n^2$. 
\end{proof}

Next, we proceed to obtain the solution of the fNLS equation \eqref{eq:fnls} through reduction. We start from a specific kernel. Let
\[
    H=L^2(\Gamma,d\nu),
    \qquad
    (Pf)(u)=u f(u),
    \qquad
    Q=P^{-1}.
\]
Since $\Gamma\subset\mathbb C\setminus\{0\}$, both $P$ and $P^{-1}$
are bounded and invertible. Choose functions $a,b$ on $\Gamma$ such that the integral operator
$\mathcal K_0$ with kernel
\begin{align}
    K_0(u,v)=\frac{a(u)b(v)}{u-v^{-1}},\label{eqn-k0}
\end{align}
is trace class. The kernel of the operator $P\mathcal K_0-\mathcal K_0Q$ is then given by
\[
    (P\mathcal K_0-\mathcal K_0Q)(u,v)
    =(u-v^{-1})K_0(u,v)
    =a(u)b(v),
\]
and hence
\[
    \operatorname{rank}(P\mathcal K_0-\mathcal K_0Q)\leq1.
\]
Thus, the assumption \eqref{eq:fnls-mother-rank} of Lem. \ref{lem:fnls-mother} is satisfied. With the definition of $\mathcal K^{(n)}$ in Eq. \eqref{eq:fnls-mother-tau}, we can obtain
\begin{equation}
 K^{(n)}(u,v)
   =(uv)^n\frac{a(u)b(v)}{u-v^{-1}}
   \exp\!\left\{
      x_{-1}(u^{-1}-v)
      +x_1(u-v^{-1})
      +x_2(u^2-v^{-2})
   \right\},
\label{eq:fnls-mkp-kernel}
\end{equation}
and consequently, the Fredholm determinant $\tau_n=\det(I+\mathcal K^{(n)})$ satisfies both the first mKP equation \eqref{eq:mkp-first} and the 2DTL equation \eqref{eq:2dtl-central}.

Note the shift property of the kernel:
\begin{align}
\mathcal K^{(n)}(x_1+c,x_{-1}+c,x_2)
=
e^{c(P+P^{-1})}\mathcal K^{(n)}(x_1,x_{-1},x_2)
e^{-c(P+P^{-1})}.
\end{align}
Since $P+P^{-1}$ is bounded, we obtain
\begin{align}
\tau_n (x_1+c,x_{-1}+c,x_2)&=\det(I+\mathcal K^{(n)}(x_1+c,x_{-1}+c,x_2))\notag\\
&=\det(I+e^{c(P+P^{-1})}\mathcal K^{(n)}(x_1,x_{-1},x_2)e^{-c(P+P^{-1})})\notag\\
&=\det(I+\mathcal K^{(n)}(x_1,x_{-1},x_2))=\tau_n (x_1,x_{-1},x_2),
\end{align}
which gives
\begin{equation}
 (\partial_{x_1}+\partial_{x_{-1}})\tau_n(x_1,x_{-1},x_2)=\partial_c\tau_n (x_1+c,x_{-1}+c,x_2)|_{c=0}=0.
 \label{eq:fnls-hierarchy-constraint}
\end{equation}
Applying Eq. \eqref{eq:fnls-hierarchy-constraint} and 
\begin{equation}
 x_{-1}=0,\qquad x_1=x,\qquad x_2=-\mathrm it,
\end{equation}
to Eq. \eqref{eq:mkp-first} and Eq. \eqref{eq:2dtl-central} yields
\begin{equation}
 (D_x^2-\mathrm iD_t)\tau_1\cdot\tau_0=0,\qquad
 (D_x^2+2)\tau_0\cdot\tau_0=2\tau_1\tau_{-1}.
 \label{eq:fnls-hierarchy-bilinear}
\end{equation}
Correspondingly, the kernel \eqref{eq:fnls-mkp-kernel} simplifies to
\begin{align}
     K^{(n)}(u,v)
   =(uv)^n\frac{a(u)b(v)}{u-v^{-1}}
   \exp\!\left\{
      x(u-v^{-1})
      -\mathrm{i}t(u^2-v^{-2})
   \right\},
\end{align}
which satisfies $K^{(n)}(\iota(u),\iota(v))=\overline{K^{(-n)}(v,u)}$ with $\iota(u)=-\frac{1}{\overline{u}}$.
Assume in addition that $\Gamma$ and the measure $\nu$ is invariant under $\iota$, we have $\tau_1=\overline{\tau_{-1}}$.  Finally, the dependent variable transformation $\phi=e^{-2\mathrm i t}{\tau_1}/{\tau_0}$ maps the bilinear system \eqref{eq:fnls-hierarchy-bilinear} directly to Eq. \eqref{eq:fnls}. Thus, the Fredholm determinant solutions to the fNLS equation can be obtained, which we summarize in the following corollary.

\begin{corollary}[Fredholm determinant solutions to the fNLS equation]
\label{cor:fnls}
Let $\Gamma\subset\mathbb C\setminus\{0\}$ be compact and invariant
under $\iota(u)=-1/\overline u$, and let $\nu$ be a finite positive
measure invariant under $\iota$.
Define $\mathcal K^{(n)}$ on $L^2(\Gamma,d\nu)$ by the kernels
\begin{align}
K^{(n)}(u,v;x,t)
=
-\frac{(uv)^n}{u-v^{-1}}
\exp\!\left[
(u-v^{-1})x-\mathrm i(u^2-v^{-2})t
\right],
\qquad n=-1,0,1.\label{kernel-fnls}
\end{align}
and put
\[
\tau_n=\det(I+\mathcal K^{(n)}),
\qquad n=-1,0,1.
\]
Then, $\tau_0$ is real and $\tau_{-1}=\overline{\tau_1}$.
The function
\[
\phi(x,t)=e^{-2\mathrm i t}\frac{\tau_1}{\tau_0}
\]
solves the fNLS equation \eqref{eq:fnls}
wherever $\tau_0\ne0$.
\end{corollary}

\begin{example}[Breather and breather gas]\label{ex:nls-breather-gas}
Let $\mathcal K^{(n)}$ on $L^2(\Gamma,d\nu)$ is defined by Eq. \eqref{kernel-fnls}. $\Gamma$, $\nu$ satisfy the assumptions in Cor. \ref{cor:fnls}.
For continuous spectral data, take
\begin{align*}
\Gamma_+
=\{\kappa e^{\mathrm i\vartheta}:a\le\kappa\le b\},
\qquad
\Gamma_-=\iota(\Gamma_+),
\qquad
\Gamma=\Gamma_+\cup\Gamma_-,
\end{align*}
where $1<a<b$ and $|\vartheta|<\pi/2$.
Prescribe $d\nu=\rho(\kappa)\,d\kappa$ on $\Gamma_+$,
with $\rho\ge0$ integrable, and use its pushforward on $\Gamma_-$.
The choices $\vartheta=0$ and $\vartheta\ne0$ give the
Kuznetsov-Ma (KM) and Tajiwara-Watanabe (TW) breather spectral families, respectively. If we take instead
\[
\Gamma_+
=\{e^{\mathrm i\vartheta}:
  \vartheta_-\le\vartheta\le\vartheta_+\},
\qquad
0<\vartheta_-<\vartheta_+<\frac{\pi}{2},
\]
with $d\nu=\rho(\vartheta)\,d\vartheta$ on $\Gamma_+$, we obtain the Akhmediev-type breather gas.
Discrete measures give finite $N$-breather solutions. Fig. \ref{nls-breathergas} shows the dynamics of KM-type breather gas, as well as the acceleration of an individual breather by the breather gas. Other results about the breather gas of fNLS equation can be found in \cite{eltovbis2020,falquigravapuntini2025,weng2025,hanzhangdong2026}.
\end{example}

\begin{remark}
    Under the transformation $\lambda=u+u^{-1},\ \mu=v+v^{-1}$, the kernel $K^{(n)}(u,v;x,t)$ can be written into the IIKS form.
\end{remark}
\begin{figure}[H]
	\centering
	\subfigure[breather gas]
    {
\includegraphics[width=2.5in]{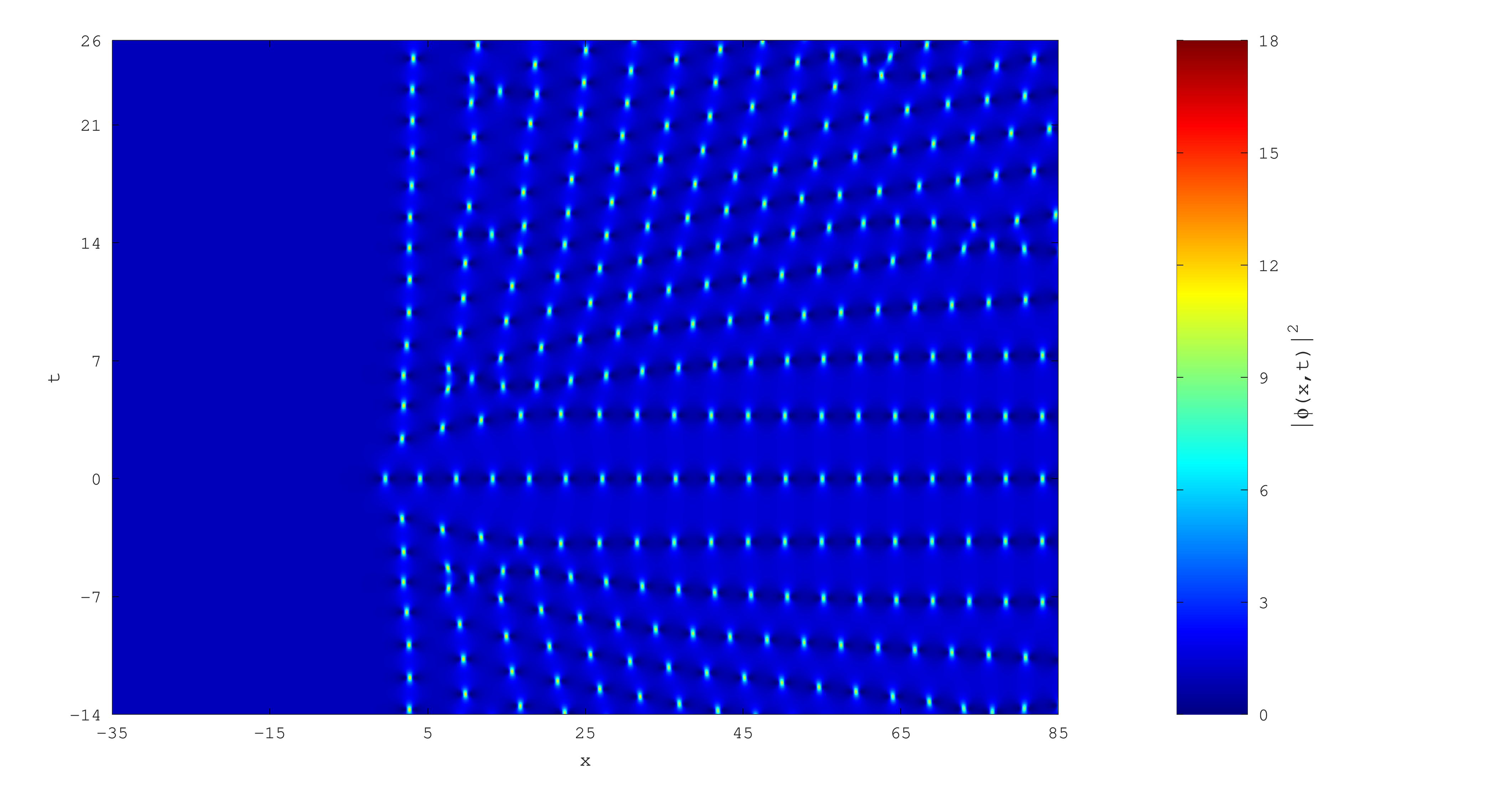}
    }
	\hspace*{3em}
    	\subfigure[breather-breather gas interaction]
	{
		\includegraphics[width=2.5in]{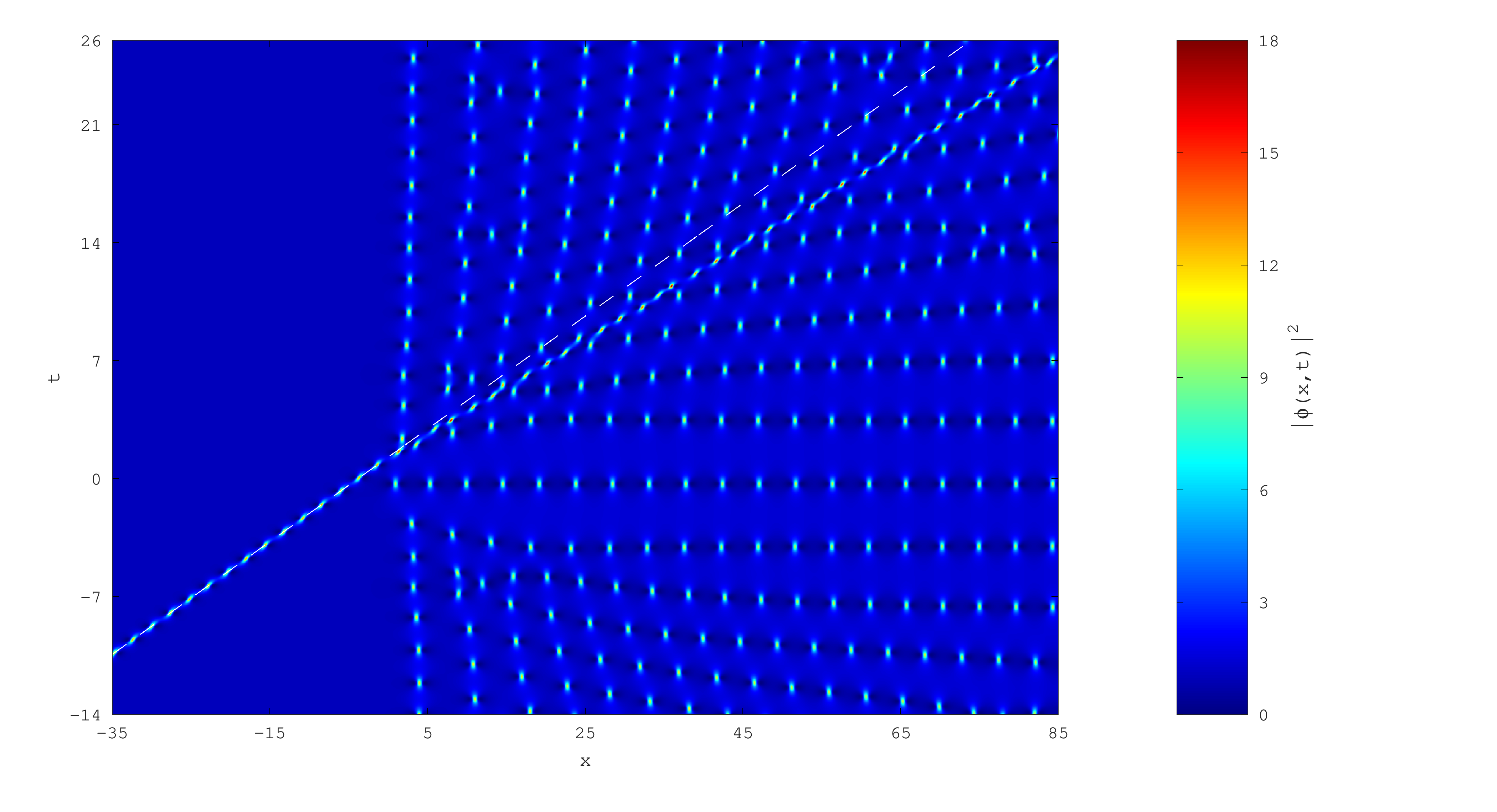}
	}
	\caption{(a) KM-type breather gas to the fNLS equation in Example \ref{ex:nls-breather-gas} with $\Gamma_+=[1.2,1.9]$; (b) interaction between single TW breather and KM-type breather gas with $\Gamma_+=\{2.6e^{-\frac{\mathrm{i}}{2}}\}\cup[1.2,1.9]$ (dash line: single TW breather without gas).}\label{nls-breathergas}
\end{figure}

\subsection{Relationship between Lem. \ref{lem:fnls-mother} and Thm. \ref{thm:main}}\label{sec:lem-thm-rel}
 We now show that, under a suitable spectral separation condition, the Fredholm determinants in Lem. \ref{lem:fnls-mother} can be represented by the half-line construction of Thm. \ref{thm:main}. Consider the operator $\mathcal{K}^{(n)}$ defined in the assumptions in Lem. \ref{lem:fnls-mother}. Suppose additionally that, for some $\delta>0$,
\begin{equation}
 \sigma(P)\subset\{\operatorname{Re}z>\delta\},\qquad
 \sigma(Q)\subset\{\operatorname{Re}z<-\delta\}.
 \label{eq:fnls-halfline-separation}
\end{equation}
As before, we write $P\mathcal K_0-\mathcal K_0Q=a_*\otimes b_*$, where $a_*$
is a vector and $b_*$ is a continuous linear functional, with
$(a_*\otimes b_*)h=a_*b_*(h)$. Set
\[
 a_n=P^ne^{\theta(P)}a_*,\qquad
 b_n=b_*e^{-\theta(Q)}Q^{-n}.
\]
The spectral separation condition \eqref{eq:fnls-halfline-separation} ensures the exponential decay of the operator $e^{rP}\mathcal K^{(n)}e^{-rQ}$ as $r\to-\infty$. By differentiating this operator with respect to $r$ and integrating from $-\infty$ to $0$, we obtain the factorization
\begin{equation}
 \mathcal K^{(n)}
 =\int_{-\infty}^0e^{rP}a_n\otimes b_ne^{-rQ}\,dr
 =A_nB_n,
 \label{eq:fnls-sylvester-factorization}
\end{equation}
where $A_n:L^2(\mathbb R_-)\to H$ and
$B_n:H\to L^2(\mathbb R_-)$ are the Hilbert-Schmidt operators
\[
 A_nh=\int_{-\infty}^0e^{rP}a_nh(r)\,dr,\qquad
 (B_nf)(u)=b_ne^{-uQ}f.
\]
Let $\mathcal F^{(n)}=(B_nA_n)^T$, where $T$ denotes kernel
transposition without conjugation. Its integral kernel is explicitly given by
\begin{equation}
 F^{(n)}(u,v)=b_ne^{-vQ}e^{uP}a_n,\qquad u,v\le0.
 \label{eq:fnls-associated-halfline-kernel}
\end{equation}
Applying Sylvester's determinant identity, we have
\begin{equation}
 \det\nolimits_H(I+\mathcal K^{(n)})
 =\det\nolimits_{L^2(\mathbb R_-)}(I+B_nA_n)
 =\det\nolimits_{L^2(\mathbb R_-)}(I+\mathcal F^{(n)}).
 \label{eq:fnls-sylvester-identification}
\end{equation}
Direct differentiation and zero-normalized integration give
\begin{align*}
 &F^{(n)}_{x_j}
 =\bigl(\partial_u^j-(-\partial_v)^j\bigr)F^{(n)},
 j=1,2,\\
 &F^{(n)}_{x_{-1}}
 =(\partial_u^{-1}+\partial_v^{-1})F^{(n)},\\
 &\partial_uF^{(n)}+\partial_vF^{(n+1)}=0.
\end{align*}
The exponential decay as $r \to -\infty$ guarantees that the necessary primitives and parameter derivatives are trace-class, and that all boundary terms vanish. Thus, under the separation condition \eqref{eq:fnls-halfline-separation}, the conclusion follows directly from Thm. \ref{thm:main}. 

We note, however, that the fNLS reduction $Q=P^{-1}$ does not satisfy this half-line condition. Because $\operatorname{Re}(z^{-1})=\operatorname{Re}(z)/|z|^2$, the real parts of the spectra of $P$ and $P^{-1}$ share the same sign, making the separation condition \eqref{eq:fnls-halfline-separation} impossible. Consequently, the fNLS solutions constructed in Cor. \ref{cor:fnls} cannot be obtained via a direct reduction of Thm. \ref{thm:main} within the half-line class considered above.

\section{Non-isospectral Lax interpretations}\label{non-isospectral}
This appendix presents non-isospectral Lax interpretations for the equation satisfied by the finite-temperature deformed Bessel kernel determinant in Eq. \eqref{n-nkdv1}.

Let $r=r(x,t)$ and $b=b(x,t)$ be real-valued functions.  We start
from the following Lax pair
\begin{equation}
    \boldsymbol{\Upsilon}_x=U(\lambda)\boldsymbol{\Upsilon},
    \qquad
    \boldsymbol{\Upsilon}_t=V(\lambda)\boldsymbol{\Upsilon},
    \label{eq:nkdv-lax-system}
\end{equation}
where
\begin{equation}
    U(\lambda)=
    \begin{pmatrix}
        0&1\\
        \lambda-r&0
    \end{pmatrix}
    \label{eq:nkdv-U}
\end{equation}
and
\begin{equation}
    V(\lambda)=\frac{1}{4\lambda}
    \begin{pmatrix}
        b_x&-2b\\
        b_{xx}-2(\lambda-r)b&-b_x
    \end{pmatrix}.
    \label{eq:nkdv-V}
\end{equation}
In the classical (isospectral) setting, the spectral parameter $\lambda$ is independent of $x$ and $t$. Then, the compatibility
condition $\boldsymbol{\Upsilon}_{xt}=\boldsymbol{\Upsilon}_{tx}$ yields the following system
\begin{align}
    &b_{xxx}+2r_xb+4rb_x=0,                          \label{eq:negative-kdv-2}\\
    &r_t=b_x.                                         \label{eq:negative-kdv-1}
\end{align}
which is the negative KdV equation studied in \cite{Verosky1991}. 

We now introduce a non-isospectral deformation by allowing the spectral parameter to depend explicitly on the time variable. More precisely, we set
\begin{equation}
    \lambda=\lambda(z,t)=-z+t,
    \label{eq:nkdv-nonisospectral-law}
\end{equation}
where $z$ is a new spectral parameter independent of $x$ and $t$. Expressed in terms of $z$, the Lax pair \eqref{eq:nkdv-lax-system} becomes 
\begin{align}    
\boldsymbol{\Upsilon}_x=\widetilde{U}(z)\boldsymbol{\Upsilon},
    \qquad
    \boldsymbol{\Upsilon}_t=\widetilde{V}(z)\boldsymbol{\Upsilon},
    \label{eq:nkdv-lax-system-deform}
\end{align}
with
\begin{equation}
    \widetilde{U}(z)=
    \begin{pmatrix}
        0&1\\
        -z+t-r&0
    \end{pmatrix},\quad \widetilde{V}(z)=\frac{1}{4(-z+t)}
    \begin{pmatrix}
        b_x&-2b\\
        b_{xx}+2(z-t+r)b&-b_x
    \end{pmatrix}.
    \label{eq:nkdv-UV-deform}
\end{equation}
The compatibility
condition $\boldsymbol{\Upsilon}_{xt}=\boldsymbol{\Upsilon}_{tx}$ now gives
\begin{align}
   & b_{xxx}+2r_xb+4rb_x=0,                          \label{eq:negative-kdv-noniso-2}\\
     &r_t=b_x+1.                                     \label{eq:negative-kdv-noniso-1}
\end{align}
Setting $b=-2\phi_t$ and $r=-2\phi_x+t$, the relation \eqref{eq:negative-kdv-noniso-1} is always satisfied, while Eq. \eqref{eq:negative-kdv-noniso-2} reduces directly to equation \cref{n-nkdv1}. We also note that the deformed Lax pair \eqref{eq:nkdv-lax-system-deform} coincides with that given in \cite[Prop. 4.1]{ruzza2025}.
\end{document}